\documentclass[11pt,letterpaper]{article}

\usepackage[utf8]{inputenc}
\usepackage[T1]{fontenc}
\usepackage{lmodern}
\usepackage[DIV=11]{typearea} % large value = less margin (roughly speaking), recommended values: 10pt: 8, 11pt: 10, 12pt: 12

\usepackage{microtype}

\usepackage{amssymb}
\usepackage{amsmath}
\usepackage{amsthm}
\usepackage{thmtools}
\usepackage{thm-restate}

\usepackage[procnumbered,ruled,vlined,linesnumbered]{algorithm2e}

\usepackage{xcolor}
\usepackage{xspace}
\usepackage{xfrac}

\usepackage{comment}

\usepackage[backend=biber, style=alphabetic, backref=true, doi=false, url=false, maxcitenames=3, mincitenames=3, maxbibnames=10, minbibnames=10, sortlocale=en_US]{biblatex}

\usepackage[ocgcolorlinks]{hyperref} % Option ocgcolorlinks makes links only colored on screen and not on printouts
\usepackage{cleveref}

\DeclareUnicodeCharacter{221A}{$\sqrt{}$}

\colorlet{DarkRed}{red!50!black}
\colorlet{DarkGreen}{green!50!black}
\colorlet{DarkBlue}{blue!50!black}

\hypersetup{
	linkcolor = DarkRed,
	citecolor = DarkGreen,
	urlcolor = DarkBlue,
	bookmarksnumbered = true,
	linktocpage = true
}

\declaretheorem[numberwithin=section]{theorem}
\declaretheorem[numberlike=theorem]{lemma}

\declaretheorem[numberlike=theorem]{definition}
\declaretheorem[numberlike=theorem]{claim}

\declaretheorem[numberlike=theorem]{remark}

\DontPrintSemicolon
\SetKw{KwAnd}{and}
\SetProcNameSty{textsc}
\SetFuncSty{textsc}

\newcommand{\dist}{\mathbf{dist}}
\newcommand{\polylog}{\text{polylog}}

\newcommand{\adag}{\mathcal{A}_{DAG}}
\newcommand{\astardag}{\mathcal{A}^{*}_{DAG}}
\newcommand{\aato}{\mathcal{A}_{ATO}}

\title{Near-Optimal Decremental SSSP in Dense Weighted Digraphs}
\author{
Aaron Bernstein\\ Rutgers University New Brunswick \\ bernstei@gmail.com
\and
Maximilian Probst Gutenberg\\ BARC, University of Copenhagen\\ maximilian.probst@outlook.com 
\and
Christian Wulff-Nilsen \\ BARC, University of Copenhagen \\ koolooz@di.ku.dk
}

\date{}

\hypersetup{
	pdftitle = {DiSSSP},
	pdfauthor = {Aaron Bernstein, Maximilian Probst Gutenberg, Christian Wulff-Nilsen}
}

\begin{document}
\maketitle
\thispagestyle{empty}
\setcounter{page}{0}

\begin{abstract}
In the decremental Single-Source Shortest Path problem (SSSP), we are given a weighted directed graph $G=(V,E,w)$ undergoing edge deletions and a source vertex $r \in V$; let $n = |V|, m = |E|$ and $W$ be the aspect ratio of the graph. The goal is to obtain a data structure that maintains shortest paths from $r$ to all vertices in $V$ and can answer distance queries in $O(1)$ time, as well as return the corresponding path $P$ in $O(|P|)$ time.

This problem was first considered by Even and Shiloach [JACM'81], who provided an algorithm with total update time $O(mn)$ for unweighted undirected graphs; this was later extended to directed weighted graphs [FOCS'95, STOC'99].
There are conditional lower bounds showing that $O(mn)$ is in fact near-optimal [ESA'04, FOCS'14, STOC'15, STOC'20]. In a breakthrough result, Forster et al. showed that total update time $\min\{m^{7/6}n^{2/3+o(1)},m^{3/4}n^{5/4+o(1)}\}  \polylog (W) = mn^{0.9+o(1)}\polylog W$ is possible if the algorithm is allowed to return $(1+\epsilon)$-approximate paths, instead of exact ones [STOC’14, ICALP’15]. No further progress was made until Probst Gutenberg and Wulff-Nilsen [SODA'20] provided a new approach for the problem, which yields total time\\ $\tilde{O}(\min\{m^{2/3}n^{4/3}\log W, (mn)^{7/8} \log W\}) = \tilde{O}(\min\{n^{8/3}\log W, mn^{3/4} \log W\})$.

Our result builds on this recent approach, but overcomes its limitations by introducing a significantly more powerful abstraction, as well as a different core subroutine. Our new framework yields a decremental $(1+\epsilon)$-approximate SSSP data structure with total update time $\tilde{O}(n^2 \log^4 W /\epsilon)$. Our algorithm is thus \emph{near-optimal} for dense graphs with polynomial edge-weights. Our framework can also be applied to sparse graphs to obtain total update time $\tilde{O}(mn^{2/3} \log^3 W / \epsilon)$. Combined, these data structures dominate all previous results. Like all previous $o(mn)$ algorithms that can return a path (not just a distance estimate), our result is randomized and assumes an oblivious adversary.

Our framework effectively allows us to reduce SSSP in general graphs to the same problem in directed acyclic graphs (DAGs). We believe that our framework has significant potential to influence future work on directed SSSP, both in the dynamic model and in others.
\end{abstract}
\newpage

\section{Introduction}

In the Single-Source Shortest Paths (SSSP) problem, the input is a directed weighted graph $G=(V,E,w)$ and a dedicated source vertex $r \in V$, and the goal is to compute shortest paths from $r$ to every other vertex $v$ in $V$. Let $n = |V|$, $m = |E|$, and $W$ be the aspect ratio of the graph, which is the ratio of maximum to minimum edge weight. The problem can be solved in $O(m + n\log(n))$ time using Dijkstra's algorithm. In this article, we study the dynamic version of the problem, where the graph changes over time. The most general model is the \emph{fully dynamic} one, where the graph is subject to a sequence of edge insertions and deletions. Unfortunately, there are extremely strong conditional lower bound for this model \cite{roditty2004dynamic, abboud2014popular,henzinger2015unifying, gutenberg2020incrSSSP}.

For this reason, much of the research on this problem has focused on the \emph{decremental} setting. Formally, the algorithm is given a graph $G = (V,E,w)$ subject to a sequence of edge deletions and edge weight increases, and the goal is is to maintain shortest distances/paths from $r$ to every $v \in V$. Although decremental SSSP only applies to a more restricted model, it is an extremely common subroutine in other dynamic algorithms (including fully dynamic ones), and has recently been used to make progress on long-standing \emph{static} problems such as computing max-flow \cite{chuzhoy2019new, chuzhoy2019deterministic}, multi-commodity flow \cite{madry2010faster}, expanders \cite{chuzhoy2019deterministic}, and sparse cuts \cite{chuzhoy2019deterministic}. For this reason, decremental SSSP is one of the most well-studied problems in dynamic algorithms \cite{shiloach1981line, bernstein2011improved, henzinger2014decremental, henzinger2014sublinear, henzinger2015improved, henzinger2016dynamic, bernstein2016deterministic, bernstein2017deterministic, bernstein2017deterministicWeighted, chuzhoy2019new, chuzhoy2019deterministic, gutenberg2020decremental, gutenberg2020deterministic, detDiSSSP}.  %spanners and online matchings \cite{bernstein2018online}.

The first algorithm for decremental SSSP was the Even and Shiloach tree \cite{shiloach1981line}, which dates back to 1981 and has total update time $O(mn)$ over the entire sequence of deletions; it was later extended to directed weighted graphs \cite{henzinger1995fully,King99}. $O(mn)$ is conditionally optimal for the exact version \cite{roditty2004dynamic,abboud2014popular,henzinger2015unifying, gutenberg2020incrSSSP}, but one can do better with a $(1+\epsilon)$-approximation. In fact, recent research culminated in a near-optimal algorithm for \emph{undirected} graphs \cite{henzinger2014decremental} with total update time $m^{1+o(1)}\log(W)$\footnote{In the decremental setting $m$ is taken to be the number of edges in the \emph{initial} graph $G$}; there has also been more recent work improving upon $O(mn)$ with adaptive or even deterministic algorithms (see e.g. \cite{bernstein2016deterministic, bernstein2017deterministic, bernstein2017deterministicWeighted,chuzhoy2019new,chuzhoy2019deterministic,gutenberg2020deterministic, bernstein2020fully}). 

In directed graphs, however, the decremental SSSP problem remains poorly understood. The first algorithms to improve upon the classic $O(mn)$ bound (with a $(1+\epsilon)$-approximation) were by Henzinger, Forster and Nanongkai \cite{henzinger2014sublinear} in 2014; the total update time is $\min\{m^{7/6}n^{2/3+o(1)},m^{3/4}n^{5/4+o(1)}\}  \polylog (W) = mn^{0.9+o(1)} \polylog (W)$ \cite{henzinger2015improved}. Since then, the only progress on the problem is a very recent algorithm of Probst Gutenberg and Wulff-Nilsen \cite{gutenberg2020decremental} with total update time $\tilde{O}(\min \{ m^{2/3}n^{4/3}, mn^{3/4}\} \log (W) )$, or slightly better in unweighted graphs. Very recently, it was further shown in \cite{detDiSSSP} that a total update time of  $n
^{2+2/3+o(1)}$ can even be obtained deterministically. Besides this recent progress, state-of-the-art algorithms for directed graphs still lag far behind the ones for undirected ones, and only achieve small improvements beyond $O(mn)$.

For more related work on the problems of maintaining Single-Source Reachability, Strongly-Connected Components, Single-Source Shortest Paths and All-Pairs Shortest Paths in incremental, decremental and fully-dynamic graphs, we refer the reader to \Cref{sec:relatedWork}.

\paragraph{Our Contribution.} We make significant progress on this problem and present the first \emph{near-optimal} algorithm for decremental SSSP in directed dense graphs. 
%for any choice of $\epsilon$. 

\begin{theorem}
\label{thm:ContributionSSSPResult}
Given a decremental input graph $G=(V,E,w)$ with $n = |V|, m=|E|$ and aspect ratio $W$, a dedicated source $r \in V$ and $\epsilon > 0$, there is a randomized algorithm that maintains a distance estimate $\widetilde{\mathbf{dist}}(r,x)$, for every $x \in V$, such that
\[
    {\mathbf{dist}}_G(r,x) \leq \widetilde{\mathbf{dist}}(r,x) \leq (1+\epsilon) {\mathbf{dist}}_G(r,x)
\]
at any stage w.h.p. The algorithm has total expected update time $\tilde{O}(n^2 \log^4 W/\epsilon)$. Distance queries are answered in $O(1)$ time, and a corresponding path $P$ can be returned in $O(|P|)$ time.
\end{theorem}

\noindent We also present the currently fastest algorithm for sparse graphs. Combined, our two results significantly improve upon all previous work.

\begin{theorem}
\label{thm:ContributionSSSPSparseResult}
Given a decremental input graph $G=(V,E,w)$ with $n = |V|, m=|E|$ and aspect ratio $W$, a dedicated source $r \in V$ and $\epsilon > 0$, there is a randomized algorithm that maintains a distance estimate $\widetilde{\mathbf{dist}}(r,x)$, for every $x \in V$, such that
\[
    {\mathbf{dist}}_G(r,x) \leq \widetilde{\mathbf{dist}}(r,x) \leq (1+\epsilon) {\mathbf{dist}}_G(r,x)
\]
at any stage w.h.p. The algorithm has total expected update time $\tilde{O}(mn^{2/3} \log^3 W/\epsilon)$.  Distance queries are answered in $O(1)$ time, and a corresponding path $P$ can be returned in $O(|P|)$ time.
\end{theorem}

\paragraph{Adaptive Versus Oblivious Adversaries} Both our results are randomized, and assume an oblivious adversary whose update sequence does not depend on the distance estimates and query-paths returned by the algorithm. This assumption is also necessary for all previous algorithms in directed graphs that go beyond the classic $O(mn)$ bound of Even and Shiloach \cite{henzinger2014sublinear,henzinger2015improved,gutenberg2020decremental}, and achieving such a result adaptively remains a major open problem. (The result of \cite{gutenberg2020decremental} allows for adaptive \emph{distance} queries, but still assumes the adversary is oblivious to the \emph{paths} returned by the algorithm.) 
More generally, achieving truly fast adaptive algorithms seems implausible until we understand the problem well enough to at least have fast oblivious algorithms, and our result makes significant progress on this front. In fact, arguably more than any other dynamic graph algorithm, our $\tilde{O}(n^2)$ result highlights the gap between these two models: if one could design an \emph{adaptive} algorithm that matches this efficiency, then plugging it into the existing framework of \cite{chuzhoy2019new} (based in turn on the multiplicative-weight update method as described in \cite{fleischer2000approximating, garg2007faster, madry2010faster}) would yield an $\tilde{O}(n^2)$ algorithm for \emph{static} directed maximum flow, which would constitute an enormous breakthrough in the field of graph algorithms. 

\paragraph{Techniques}
The key contribution of our paper is a general technique for converting algorithms on directed acyclic graphs (DAGs) into algorithms on general graphs. Earlier techniques in \cite{bernstein2017deterministicWeighted,gutenberg2020decremental} lead to two simple algorithms for DAGs with total update times $\tilde{O}(n^2)$ and $\tilde{O}(mn^{2/3})$; our conversion then extends these bounds to general graphs. We first introduce the concept \emph{approximate topological order} ($\mathcal{ATO}$), which loosely speaking imposes a DAG-like structure on \emph{any} graph. We then show that an $\mathcal{ATO}$ always exists and can be maintained efficiently.

At a high-level, the conversion is as follows. Let $\adag$ be a decremental SSSP algorithm for DAGs and let $T(\adag)$ be the total update time. The first (easier) step is to convert $\adag$ to an algorithm $\astardag$ that works on any graph with an $\mathcal{ATO}$ and has total update time $T(\astardag) \sim T(\adag)$. The second (harder) step is 
to build an algorithm $\aato$ that maintains an $\mathcal{ATO}$ in $G$ by recursively applying $\astardag$ as a subroutine: the basic idea is that an $\mathcal{ATO}$ of better ``quality" can be built by using $\astardag$ to maintain shortest paths in an $\mathcal{ATO}$ of worse quality. Using this layered approach, we can achieve $T(\aato) \sim T(\astardag) \sim T(\adag)$, and combining $\aato$ with $\astardag$ gives an algorithm for general graphs. 

Neither the step from $\adag$ to $\astardag$ nor the step from $\astardag$ to $\aato$ are black-box, but the techniques are quite modular and flexible, as evidenced by the fact that we were able to apply this conversion to both of the state-of-the-art algorithms for DAGs. We believe our conversion has strong potential to influence future work on directed shortest paths, in both the dynamic model and in others, by allowing researchers to focus on the simpler case of DAGs.

%\paragraph{Techniques.} Our algorithm builds on the techniques developed in \cite{gutenberg2020decremental}, however, new key ideas are necessary to make our algorithm work. At the heart of our algorithm is the new notion of an "approximate" topological order, a function $\tau$ that allows us to straight-forwardly extend a fast data structure for directed acyclic graphs (DAGs) to work on general graphs. It is therefore no surprise that the bounds derived for dense and sparse graphs respectively are obtained by using the state of the art algorithms for solving the problems on DAGs. This abstraction might also be useful in other contexts related to shortest paths where solving the problem on DAGs is easier than on general graphs.

%In order to maintain this function, we use a separator procedure that chooses a random layer according to the exponential distribution. However, in comparison to another prominent algorithm that features the exponential distribution, the random-shift algorithm \cite{miller2013parallel}, we use the exponential distribution in a more straight-forward way to select a random layer from $D$ layers such that we select layers close to each other with roughly similar probability but select the final $\frac{D}{ c \log n}$ layers with probability at most $1/n^c$ for some large constant $c$. We believe that this technique might be of further interest in the design of decremental algorithms. 

\section{Preliminaries}
\label{sec:prelim}

We let a graph $H$ refer to a weighted, directed graph with vertex set denoted by $V(H)$, edge set $E(H)$ and weight function $w_H : E(H) \rightarrow [1,W] \cup \{\infty\}$\footnote{In this article, we restrict our attention to integers and let $[a,b]$ denote the set of integers $a, a+1, a+2, \dots, b$.}. We say that $H$ is a \emph{decremental} graph if it is undergoing a sequence of edge deletions and edge weight increases (also referred to as updates), and refer to \emph{version} $t$ of $H$, or $H$ at \emph{stage} $t$ as the graph $H$ obtained after the first $t$ updates have been applied. In this article, we denote the (decremental) input graph by $G=(V,E,w)$ with $n = |V|$ and $m = |E|$ (where $m$ refers to the number of edges of $G$ at stage $0$). In all subsequent definitions, we often use a subscript to indicate which graph we refer to, however, when we refer to $G$, we often omit the subscript.

\paragraph{Cuts, Neighborhoods and Subgraphs.} For graph $H$, and any two disjoint subsets $X,Y \subseteq V(H)$, we let $E_H(X)$ be the set of edges in $E(H)$ with an endpoint in $X$, and $E_H(X,Y)$ denote the set of edges in $E(H)$ with tail in $X$ and head in $Y$. We let $\mathcal{N}^{in}_H(v) = \{ u \;|\; (u,v) \in E(H) \}$ and $\mathcal{N}^{out}_H(v)  = \{ u \;|\; v \in \mathcal{N}^{in}_H(u)\}$ denote the in-neighborhood and out-neighborhoods of $v \in V$. We further let $H[X]$ refer to the subgraph of $H$ induced by $X$, i.e. $H[X] = (X, E_H(X,X), w_H)$. We use $H \subseteq G$ to denote that $V(H) = V(G)$ and $E(H) \subseteq E(G)$.

\paragraph{Contractions.} We define the graph $H / X$ to be the graph obtained from $H$ by contracting all vertices in $X$ into a single node (we use the word \emph{node} instead of vertex if it was obtained by contractions). Similarly, for a set of pairwise disjoint vertex sets $X_1, X_2, \dots, X_k$, we let $H / \{X_1, X_2, \dots, X_k\}$ denote the graph $((((H / X_1) / X_2) \dots ) / X_k)$. If $\mathcal{V}$ forms a partition of $V$, we use the convention to denote by $X^v$ the node in $G/ \mathcal{V}$ that contains $v \in V$, i.e. $v \in X^v$.

\paragraph{Reachability and Strong-Connectivity.} For graph $H$ and any two vertices $u,v \in V(H)$, we let $u \leadsto_H v$ denote that $u$ can reach $v$ in $H$, and $u \rightleftarrows_H v$ that $u$ can reach $v$ and vice versa (in the latter case, we also say $u$ and $v$ are strongly-connected). For any sets $X,Y \subseteq V(H)$, we say that $X \leadsto_H Y$ if there exists some $x \in X$, $y \in Y$ such that $x \leadsto_H y$; we define $X \rightleftarrows_H Y$ analogously. We say that the partition of $V(H)$ induced by the equivalence relation $\rightleftarrows_H$ is the set of \emph{strongly-connected components} (SCCs). 

\paragraph{Generalized Topological Order.} We define a \emph{generalized topological order} \\ $\textsc{GeneralizedTopOrder}(H)$ to be a tuple $(\mathcal{V}, \tau)$ where $\mathcal{V}$ is the set of SCCs of $H$ and $\tau : \mathcal{V} \rightarrow [0,n)$ is a function that maps any sets $X, Y \in \mathcal{V}$ such that $\tau(X) < \tau(Y)$, if $X \leadsto_H Y$ and such that $[\tau(X), \tau(X) + |X|) \cap [\tau(Y), \tau(Y) + |Y|) = \emptyset$. Thus, $\tau$ effectively establishes a one-to-one correspondence between $|X|$-sized intervals and SCCs $X$ in $H$. We point out that a $\textsc{GeneralizedTopOrder}(H)$ can always be computed in $O(|E(H)|)$ time \cite{tarjan1972depth}. In fact, a generalized topological order can also be maintained efficiently in a decremental graph $H$. Here, we say that $(\mathcal{V},\tau)$ is a dynamic tuple that forms a generalized topological order of $H$ if it is a topological order for all versions of $H$. Further, we say that $(\mathcal{V}, \tau)$ has the \emph{nesting} property, if for any set $X \in \mathcal{V}$ and a set $Y \supseteq X$ that was in $\mathcal{V}$ at an earlier stage, we have $\tau(X) \in [\tau(Y), \tau(Y) + |Y| - |X|]$; in other words, the interval $[\tau(X), \tau(X)+|X|)$ is entirely contained in the interval $[\tau(Y), \tau(Y) + |Y|)$. Thus, the associated interval with $X$ is contained in the interval associated with $Y$. We refer to the following result that can be obtained straight-forwardly by combining the data structure given in \cite{bernstein2019decremental} and the static procedure by Tarjan \cite{tarjan1972depth} as described in \cite{gutenberg2020decremental}.

\begin{theorem}[see \cite{tarjan1972depth,bernstein2019decremental, gutenberg2020decremental}]
\label{thm:SCCinDecrGraph}
Given a decremental digraph $H$, there exists an algorithm that can maintain the generalized topological order $(\mathcal{V}, \tau)$ of $H$ where $\tau$ has the nesting property. The algorithm runs in expected total update time $O(m \log^4 n)$, is randomized and works against an adaptive adversary. 
\end{theorem}

\paragraph{Distances and Diameter.} We let $\mathbf{dist}_H(u,v)$ denote the distance from vertex $u$ to vertex $v$ in graph $H$ and denote by $\pi_{u,v, H}$ the corresponding shortest path (we assume uniqueness by implicitly referring to the lexicographically shortest path). We define the weak diameter of $X \subseteq V(H)$ in $H$ by $\mathbf{diam}(X, H) = \max_{x,y \in X} \mathbf{dist}_H(x,y)$.

\paragraph{Exponential Distribution.} Finally, we make use of the exponential distribution, that is we use random variables $X$ with cumulative distribution function $F_X(x,\lambda) = 1 - e^{-\lambda x}$ for all $x \geq 0, \lambda > 0$, which we denote by the shorthand $X \sim \textsc{Exp}(\lambda)$. If $X \sim \textsc{Exp}(\lambda)$ is clear, we also use $F_X(x)$ in place of $F_X(x,\lambda)$. The exponential distribution has the special property of being \emph{memoryless}, that is if $X \sim \textsc{Exp}(\lambda)$, then
\[
    \mathbb{P}[X > s + t \;|\; X > t] = \mathbb{P}[X > s].
\]

\section{Overview}
\label{sec:overview}

We now give an overview of our algorithm. In order to illustrate the main concepts, we start by giving a simple algorithm to obtain total update time $O(n^2 \log n/\epsilon)$ in directed acyclic graphs (DAGs). While this algorithm was previously not explicitly mentioned, it follows rather directly from the techniques developed in \cite{bernstein2017deterministicWeighted, gutenberg2020decremental}. We present this algorithm to provide intuition for our approach and motivates our novel notion of approximate topological orders. In the light of approximate topological orders, we then shed light on limitations of the previous approach in \cite{gutenberg2020decremental} and present techniques to surpass these limitations to obtain \Cref{thm:ContributionSSSPResult}. In this overview, we focus on obtaining an SSSP algorithm that runs in $\tilde{O}(n^2 \log^4 W / \epsilon)$ expected update time.
The final paragraph then sketches our improvement for sparse graphs (\Cref{thm:ContributionSSSPSparseResult}); this result combines our directed framework in a non-trivial way with ideas from an earlier result for sparse SSSP in \emph{undirected, unweighted} graphs \cite{bernstein2017deterministic}.

%Finally, we briefly discuss the organization of the article.

\subsection{A Fast Algorithm for DAGs} 
\label{subsec:DAGalgo}

\paragraph{The Topological Order Difference.} Let $G = (V,E,w)$ be a DAG and let $\tau$ be the function returned by  $\textsc{GeneralizedTopOrder}(G)$ computed on the initial version of $G$ (since $G$ is a DAG, this is just a standard topological order). Let us now make an almost trivial observation: for any shortest $s$-to-$t$ path $\pi_{s,t}$, in any version of $G$, the sum of topological order differences is bounded by $n$. More formally:
\begin{equation}
\label{eq:defT}
        \mathcal{T}(\pi_{s,t}, \tau) \stackrel{\text{def}}{=} \sum_{(u,v) \in \pi_{s,t}} \tau(v) - \tau(u) = \tau(t) - \tau(s) \leq n.
\end{equation}
%since $\tau$ is strictly increasing along the vertices of the path. 
Observe that every path in $G$ can only contain few edges $(u,v)$ with large topological order difference, i.e. with $\tau(v) - \tau(u)$ large, by the pigeonhole principle.

\paragraph{Reviewing the ES-tree.} To understand how this fact can be exploited, let us review the classic ES-tree (see \cite{shiloach1981line}): We maintain distances from a fixed source $r \in V$ to each vertex $v$ in $V$ up to distance $\delta > 0$ by storing a distance estimate $\widetilde{\mathbf{dist}}(r,v)$ that is initialized to the distance between $r$ and $v$ in $G$ along with a shortest-path tree $T$ rooted at $r$. On update $(u,v)$, we delete the edge from $G$ and possibly from $T$. Then, if possible, we extract $w_{min} \in V \setminus \{r\}$, the vertex with smallest distance estimate among vertices without incoming edge in $T$. (For the first extraction, we always have $w_{min} = v$.) For this vertex $w_{min}$, we then try to find a vertex $x \in \mathcal{N}^{in}(w_{min})$ such that adding $(x, w_{min})$ to $T$ implies $\mathbf{dist}_T(r,w_{min}) \leq \widetilde{\mathbf{dist}}(r,w_{min})$. We therefore search in $\mathcal{N}^{in}(v)$ for an $x$ that satisfies \begin{equation}
\label{eq:enforce}
    \widetilde{\mathbf{dist}}(r,x)+ w(x,w_{min}) \leq \widetilde{\mathbf{dist}}(r,w_{min}).
\end{equation}
If no such $x$ exists, then $\widetilde{\mathbf{dist}}(r,w_{min})$ has to be incremented, and we set $T$ to $T \setminus N^{out}(w_{min})$. If $\widetilde{\mathbf{dist}}(r,w_{min}) > \delta$, we set it to $\infty$ and remove $w_{min}$ from the tree. We iterate the process until $T$ is spanning for vertices with distance estimate $< \infty$. Since for each distance estimate value $\widetilde{\mathbf{dist}}(r,v)$, the in-neighourhood $\mathcal{N}^{in}(v)$ has to be scanned once (in an efficient implementation), and since estimates increase monotonically, the total update time can be bound by $O(\sum_{v \in V} |\mathcal{N}^{in}(v)| \delta) = O(m \delta)$. Further, observe that if a vertex $v$ was only allowed to scan a certain vertex $x \in \mathcal{N}^{in}(v)$ every $i$ distance estimates (for example whenever $\widetilde{\mathbf{dist}}(r,v) $ is divisible by $ i$) then this corresponds to enforcing that at all times $\widetilde{\mathbf{dist}}(r,x)+ w(x,v) + i - 1 \leq \widetilde{\mathbf{dist}}(r,v)$ since we check equation \ref{eq:enforce} every $i$ steps (in particular, for $i=1$, we get an exact algorithm). Consequently, we get at most $i-1$ additive error in the distance estimate for any $t$ whose shortest $r$-to-$t$ path $\pi_{r,t}$ contains $(x,v)$. On the other hand, we only need to scan and check the edge $(x,v)$, by the argument above, $\delta / i$ times\footnote{This trade-off was first observed in \cite{bernstein2017deterministicWeighted}.}.

\paragraph{Improving the Running Time.} Let us now exploit Inequality \ref{eq:defT}. We therefore define 
\[
B_j(v) = \{ u \in \mathcal{N}^{in}(v) \textit{ with }  2^{j} \leq \tau(v) - \tau(u) < 2^{j+1}\}
\]
for every $v \in V$ and $0 \leq j \leq \lg n$; the $B_j(v)$ partition the in-neighborhood of $v$ according to topological order difference to $v$. Observe that $|B_j(v)| \leq 2^j$. Now, consider the algorithm as above where for every vertex $v$, instead of checking all $\mathcal{N}^{in}(v)$, we only check edge $(x,v)$ for $x \in B_j(v)$ if 
$\widetilde{\mathbf{dist}}(r,v)$ is divisible by $\lceil2^j\frac{\epsilon \delta}{n}\rceil$.
By the arguments above the total running time now sums to 
\[O\left(\sum_{v \in V} \sum_{0 \leq j \leq \lg n} |B_j(v)| \frac{\delta}{2^j\frac{\epsilon \delta}{n}}\right) = O\left(\sum_{v \in V} \sum_{0 \leq j \leq \lg n} 2^{j+2} \frac{\delta}{2^j\frac{\epsilon \delta}{n}}\right) = \tilde{O}(n^2/\epsilon).
\]

%with current distance estimate $\widetilde{\mathbf{dist}}(r,v)$, instead of checking all in-neighbors in $\mathcal{N}^{in}(v)$, when repairing $T$, we only check edges $(x,v)$ for vertices $x \in B_j(v)$ for $j$ where $\widetilde{\mathbf{dist}}(r,v)\; | \;\lfloor2^j\frac{\epsilon \delta}{n}\rfloor$. 

\paragraph{Bounding the Error.} Fix a shortest $r$-to-$t$ path $\pi_{r,t}$, and consider any edge $(u,v) \in \pi_{r,t}$ with $u \in B_{j+1}(v)$. We observe that the edge $(u,v)$ contributes at most an additive error of $2^{j+1}\frac{\epsilon \delta}{n}$ to $\widetilde{\mathbf{dist}}(r,t)$ since it is scanned every $\lceil 2^{j}\frac{\epsilon \delta}{n} \rceil$ distance values and if $\lceil 2^{j}\frac{\epsilon \delta}{n} \rceil$ is equal to $1$ it does not induce any error.

On the other hand, since $u \in B_{j+1}(v)$ we also have $\tau(v) - \tau(u) \geq 2^j$. We can thus charge $n/(2 \epsilon \delta)$ units from $\mathcal{T}(\pi_{r,t}, \tau)$ for each additive error unit; we know from Equation \ref{eq:defT} that $\mathcal{T}(\pi_{r,t}, \tau) \leq n$, so the total additive error is at most 
 $\frac{n}{n/(2\epsilon \delta)} = 2\epsilon\delta$. Thus, for all distances $\approx \delta$ (say in $[\delta/2, \delta)$), we obtain a $(1+4\epsilon)$-multiplicative distance estimate\footnote{Technically, we run to depth $(1+4\epsilon)\delta$ to ensure that vertices' distance estimates are not set to $\infty$ too early.}.

\paragraph{Working with Multiple Distance Scales.} Observe that the data structure above has no running time dependency on $\delta$. Thus, to obtain a data structure that maintains a $(1+2\epsilon)$-approximate distance estimate from $r$ to any vertex $x$, we can simply use $\lg(nW)$ data structures in parallel where the $i^{th}$ data structure has $\delta = 2^i$. A query can then be answered by returning the smallest distance estimate from any data structure, using a min-heap data structure to obtain this smallest estimate in constant time\footnote{Here, we exploit that all distance estimates are overestimates, and at least one of them is $(1+2\epsilon)$-approximate.}. The running time for all data structures is then bounded by $\tilde{O}(n^2 \log W /\epsilon)$.

\subsection{Extending the Result to General Graphs}
\label{subsec:ExtendSSSPToGenGraphs}

We now encourage the reader to verify that in the data structure for DAGs, we used at no point that the graph was acyclic, but rather only used that $\mathcal{T}(\pi_{s,t}, \tau)$ is bounded by $n$ for any path $\pi_{s,t}$. In this light, it might be quite natural to ask whether such a function $\tau$ might exist for general decremental graphs. Surprisingly, it turns out that after carrying out some contractions in $G$ that only distort distances slightly, we can find such a function $\tau$ that comes close in terms of guarantees. We call such a function $\tau$ an \emph{approximate topological order} (this function will no longer encode guarantees about reachability, but it helps for intuition to think of $\tau$ as being similar to a topological order). 

\paragraph{The Approximate Topological Order} We start with the formal definition:

\begin{definition}
\label{def:ATO}
Given a decremental weighted digraph $G = (V,E,w)$ and parameter $\eta_{diam} \geq 0$. We call a dynamic tuple $(\mathcal{V}, \tau)$ an \emph{approximate topological order of $G$ of quality $q > 1$} (abbreviated $\mathcal{ATO}(G, \eta_{diam})$ of quality $q$), if at any stage
\begin{enumerate}
    \item $\mathcal{V} = \{ X_1, X_2, .., X_k\}$ forms a partition of $V$ and a refinement of all earlier versions of $\mathcal{V}$, \label{prop:VUpdate} and

    \item $\tau : \mathcal{V} \rightarrow [0, n)$ is a function that maps each $X \in \mathcal{V}$ to a value $\tau(X)$. If some set $X \in \mathcal{V}$ is split at some stage into disjoint subsets $X_1, X_2, .., X_k$, then we let $\tau(X_{\pi(1)}) = \tau(X)$ and $\tau(X_{\pi(j+1)}) = \tau(X_{\pi(j)}) + |X_{\pi(j)}|$ for each $j < k$ and some permutation $\pi$ of $[1, k]$, and \label{prop:tauUpdate}
    
    \item each $X \in \mathcal{V}$ has weak diameter $\mathbf{diam}(X, G) \leq \frac{|X| \eta_{diam}}{n}$, and \label{prop:ContractLittle}
    
    \item \label{prop:TauTotal} for any two vertices $s, t \in V$, the shortest path $\pi_{s,t}$ in $G$ satisfies $\mathcal{T}(\pi_{s,t}, \tau) \leq q \cdot w(\pi_{s,t}) + n$ where we define $\mathcal{T}(\pi_{s,t}, \tau) \stackrel{\text{def}}{=} \sum_{(u,v) \in \pi_{s,t}} |\tau(X^u) - \tau(X^v)|$.
\end{enumerate}
We say that $(\mathcal{V}, \tau)$ is an $\mathcal{ATO}(G, \eta_{diam})$ of \emph{expected} quality $q$, if $(\mathcal{V}, \tau)$ satisfies properties \ref{prop:VUpdate}-\ref{prop:ContractLittle}, and at any stage, for every $s,t \in V$, $\mathbb{E}[\mathcal{T}(\pi_{s,t}, \tau)] \leq q \cdot w(\pi_{s,t}) + n$.
\end{definition}

Let us expound the ideas captured by this definition. We remind the reader that such a function $\tau$ is required by a data structure that only considers distances in $[\delta/2, \delta)$. Let us consider a tuple $(\mathcal{V}, \tau)$ that forms an $\mathcal{ATO}(G, \epsilon \delta)$ of quality $q$. Then, for any $s$-to-$t$ shortest path $\pi_{s,t} = \langle s = v_1, v_2, \dots v_{\ell} = t\rangle$ in $G$, let $s_i$ and $t_i$ be the first and last vertex on the path in $X_i \in \mathcal{V}$ (see property \ref{prop:VUpdate}) if there are any. Observe that by property \ref{prop:ContractLittle}, the vertices $s_i$ and $t_i$ are at distance at most $\frac{|X_i|\epsilon \delta}{n}$ in $G$. It follows that if we contract the SCC $X_i$, the distance $\mathbf{dist}_{G / X_i}(s,t)$ is at least the distance from $s$ to $t$ in $G$ minus an additive error of at most $\frac{|X_i|\epsilon \delta}{n}$. It follows straight-forwardly, that after contracting all sets in $\mathcal{V}$, we have that distances in $G / \mathcal{V}$ correspond to distances in $G$ up to a negative additive error of at most $\sum_{X_i \in \mathcal{V}} \frac{|X_i| \epsilon\delta}{n} = \frac{n \cdot \delta\epsilon}{n} = \epsilon \delta$. Thus, maintaining the distances in $G / \mathcal{V}$ $(1+2\epsilon)$-approximately is still sufficient for getting a $(1 \pm 2\epsilon)$-approximate distance estimate.

Property \ref{prop:VUpdate} simply ensures that the vertex sets forming the elements of $\mathcal{V}$ decompose over time. Property \ref{prop:tauUpdate} states that $\tau$ assigns each node in $G / \mathcal{V}$ a distinct number in $[0,n)$. It also ensures that if a set $X \in \mathcal{V}$ receives $\tau(X)$ that every later subset of $X$ will obtain a number in the interval $[\tau(X), \tau(X) + |X|)$. Moreover, $\tau$ effectively establishes a one-to-one correspondence between nodes in a version of $G / \mathcal{V}$ and intervals in $[0, n)$ of size equal to their underlying vertex set. Once a set $X\in\mathcal{V}$ decomposes into sets $X_1, X_2, \dots$, property \ref{prop:tauUpdate} stipulates that the intervals that $\tau$ maps $X_1, X_2, \dots$ to are disjoint subintervals of $[\tau(X), \tau(X) + |X|)$. We point out that once $\mathcal{V}$ consists of singletons, each vertex is essentially assigned a single number.

Finally, property \ref{prop:TauTotal} gives an upper bound on the topological order difference. Observe that we redefine $\mathcal{T}(\pi_{s,t}, \tau)$ in a way that is consistent with Definition \ref{eq:defT}. 
In our algorithm, for $\eta_{diam} \approx \epsilon\delta$, we obtain a quality of $\tilde{O}(n/\epsilon\delta)$. Thus, any path $\pi$ of weight $\approx \delta$ has $\mathcal{T}(\pi / \mathcal{V}, \tau) \leq \tilde{O}(n / \epsilon)$ which is very close to the upper bound obtained by the topological order function in DAGs. We summarize this result in the theorem below which is one of our main technical contributions.

\begin{restatable}{theorem}{ATOResultIntro}[see \Cref{sec:ATORealImpl} and \Cref{subsec:bootstrapDense}.]
\label{thm:TopologicalOrderMaintenanceOverview}
For any $0 \leq i \leq \lg(Wn)$, given a decremental digraph $G=(V,E,w)$, we can maintain an $\mathcal{ATO}(G, 2^i)$ of expected quality $\tilde{O}(n/2^i)$. The algorithm runs in total expected update time $\tilde{O}(n^2)$ against a non-adaptive adversary with high probability. 
\end{restatable}

Combining the theorem above and the theorem below which is obtained by generalizing the above decremental SSSP algorithm for DAGs, we obtain our main result \Cref{thm:ContributionSSSPResult}.

\begin{restatable}{theorem}{ssspSuperSimple}[see \Cref{subsec:alphaDeltaSSSPReduction}.]
\label{thm:SSSPEfficientIntro}
Given $G=(V,E,w)$ and $(\mathcal{V}, \tau)$ an  $\mathcal{ATO}(G, \eta_{diam} \approx \epsilon\delta)$, for some depth parameter $\delta > 0$, of quality $q$, a dedicated source $r$ in $V$, and an approximation parameter $\epsilon > 0$. Then, there exists a deterministic data structure $\mathcal{E}_r$ that maintains a distance estimate $\widetilde{\mathbf{dist}}(r,v)$ for each $v \in V$, which is guaranteed to be $(1+\epsilon)$-approximate if ${\mathbf{dist}}(r,v) \in [\delta, 2\delta)$. Distance queries are answered in $O(1)$ time and a corresponding path $P$ can be returned in $O(|P|)$ time. The total update time is $\tilde{O}(n \delta q/\epsilon + n^2)$.
\end{restatable}

\subsection{The Framework by \texorpdfstring{ \cite{gutenberg2020decremental}}{Gutenberg and Wulff-Nilsen}}

Before we describe our new result, we review the framework in \cite{gutenberg2020decremental} to construct and maintain an $\mathcal{ATO}(G, \eta_{diam})$. We point out that while the abstraction of an approximate topological order is new to our paper, analyzing the technique in \cite{gutenberg2020decremental} through the $\mathcal{ATO}$-lens is straight-forward and gives a first non-trivial result. Throughout this review section, we assume that the graph $G$ is unweighted to simplify presentation. This allows us to make use of the following result which states that for vertex sets that are far apart, one can find deterministically a vertex separator that is small compared to the smaller side of the induced partition (to obtain an algorithm for weighted graphs a simple edge rounding trick is sufficient to generalize the ideas presented below).

\begin{definition}
\label{def:balancedVertexSep}
Given graph $G = (V,E,w)$, then we say a partition of $V$ into sets $A, S_{Sep}, B$ is a one-way vertex separator if $A \not\leadsto_{G \setminus S_{Sep}} B$ and $A$ and $B$ are non-empty. 
\end{definition}

\begin{lemma}[see Definition 5 and Lemma 6 in \cite{chechik2016decremental}]
\label{lma:balancedVertexSep}
Given an unweighted graph $G$ of diameter $\mathbf{diam}(G)$. Then we can find sets $A, S_{Sep}, B$ that form a one-way vertex separator such that $|S_{Sep}| \leq \tilde{O}(\frac{\min\{|A|, |B|\}}{\mathbf{diam}(G)})$ in time $O(m)$. 
\end{lemma}

\paragraph{High-level Framework.} The main idea of \cite{gutenberg2020decremental} is to maintain a tuple $(\mathcal{V}, \tau)$ which is an $\mathcal{ATO}(G, \eta_{diam})$ by setting $(\mathcal{V}, \tau)$ to be $\textsc{GeneralizedTopOrder}(G')$ of some decremental graph $G' \subseteq G$ (over the same vertex set, i.e. $V(G') = V(G)$). It is straight-forward to see that $(\mathcal{V}, \tau)$ satisfies property \ref{prop:VUpdate} in \Cref{def:ATO}, since SCCs in the decremental graph $G'$ decompose. Further, it is not hard to extend the existing algorithm for maintaining SCCs in a decremental graph $G'$ given in \cite{bernstein2019decremental} to also maintain function $\tau$ that obeys property \ref{prop:tauUpdate} in \Cref{def:ATO}. The algorithm to maintain $(\mathcal{V}, \tau)$ given $G'$ runs in total update time $\tilde{O}(m)$ (same as in \cite{bernstein2019decremental}).
%that SCCs  $X_1, X_2, \dots, X_k$ obtained by decomposing a former SCC $X$, each inherit a $|X_i|$-sized disjoint subinterval of $[\tau(X), \tau(X) + |X|)$. 

So far, we have not given any reason why $G'$ needs to be a subgraph of $G$. To see why we cannot use the above strategy on $G$ directly, recall property \ref{prop:ContractLittle} in the $\mathcal{ATO}$-definition \ref{def:ATO}, which demands that each SCC $X$ has weak diameter at most $\frac{|X|\eta_{diam}}{n}$. This property might not hold in the main graph $G$. In order to resolve this issue, $G'$ is initialized to $G$ and then the diameter of SCCs in $G'$ is monitored. Whenever an SCC $X$ violates property \ref{prop:ContractLittle}, a vertex separator $S_{Sep}$ is found in the graph $G'[X]$ as described in \Cref{lma:balancedVertexSep} and all edges incident to $S_{Sep}$ are removed from $G'$. Letting $S$ denote the union of all such separators $S_{Sep}$, we can now write $G' = G \setminus E(S)$. 

\paragraph{Establishing the Quality Guarantee.} To establish a quality $q$ of the $\mathcal{ATO}$ as described in property \ref{prop:TauTotal} in \Cref{def:ATO}, let us first partition the set $S$ into sets $S_0, S_1, \dots, S_{\lg n}$ where each $S_i$ contains all separator vertices found on a graph of size $[n/2^i, n/2^{i+1})$, thus it was found when the procedure from \Cref{lma:balancedVertexSep} was invoked on a graph with diameter at least $\frac{(n/2
^{i+1})\eta_{diam}}{n} =  \frac{\eta_{diam}}{2
^{i+1}}$. Since separators are further balanced, i.e. there size is controlled by the smaller side of the induced cut, we can further use induction and \Cref{lma:balancedVertexSep} to establish that there are at most $\tilde{O}(\frac{2^i n}{\eta_{diam}})$ vertices in $S_i$. Next, observe that since any separator that was added to $S_i$ was found in a graph $G'[X]$ with $|X| \in [n/2^i, n/2^{i+1})$, we have by property \ref{prop:tauUpdate} that nodes $X' \subseteq X$ that are in the current version of $\mathcal{V}$ are assigned a $\tau(X')$ from the interval $[\tau(X), \tau(X) + |X|)$. Thus, any edge $(x,s)$ or $(s,x)$ with $x \in X, s \in S_i \cap X$ has topological order difference $|\tau(X^x) - \tau(X^s)| \leq |X| \leq n/2^i$.

Finally, let us define
\begin{equation}
\label{eq:redefT}
        \mathcal{T}'(\pi_{s,t}, \tau) \stackrel{\text{def}}{=} \sum_{(u,v) \in \pi_{s,t}} \min\{0, \tau(X^v) - \tau(X^u)\}
\end{equation}
the function similar to $\mathcal{T}(\pi_{s,t}, \tau)$ that only captures negative terms, i.e. sums only over edges that go "backwards" in $\tau$. Observe that $\mathcal{T}(\pi_{s,t}, \tau) \leq 2 \mathcal{T}'(\pi_{s,t}, \tau) + n$. Now, since $(\mathcal{V}, \tau)$ is a $\textsc{GeneralizedTopOrder}(G')$, we have that $(u,v)$ occurs in the sum of $\mathcal{T}'(\pi_{s,t}, \tau)$ if and only if $(u,v) \in G \setminus G'$, so one endpoint is in a set $S_i$ and therefore $\tau(X^v) -\tau(X^u) \leq n/2^i$. Since we only have two edges on any shortest path incident to the same vertex, we can establish that
\[
    \mathcal{T}'(\pi_{s,t}, \tau) \leq \sum_{i} 2|S_i| n/2^i = \tilde{O}(n^2/\eta_{diam}).
\]
We obtain that $(\mathcal{V}, \tau)$ is a $\mathcal{ATO}(G, \eta_{diam})$ of quality $\tilde{O}(\frac{n^2}{2^i\eta_{diam}})$ for all paths of weight at least $2^i$. Thus, when the distance scale $\delta \geq \sqrt{n}/\epsilon$, \Cref{thm:SSSPEfficientIntro} requires total update time $\approx n^{2.5}$ to maintain $(1+\epsilon)$-approximate SSSP. For distance scales where $\delta < \sqrt{n}/\epsilon$, a classic ES-tree has total update time $\approx m\sqrt{n} \leq n^{2.5}$.

%Using a classic ES-tree for distances up to $\sqrt{n}/\epsilon$ and by plugging our $\mathcal{ATO}$, for multiple scales, into \Cref{thm:SSSPEfficientIntro} gives total update time $\approx n^{2.5}$ to maintain $(1+\epsilon)$-approximate SSSP. 

\paragraph{Limitations of the Framework.} 
 Say that the goal is to  maintain shortest paths of length around $\sqrt{n}$.
The first step in the framework of \cite{gutenberg2020decremental} is to find separator $S$ such that all SCCs of $G' = G \setminus E(S)$ have diameter at most $\epsilon \sqrt{n}$ and then maintain $(\mathcal{V},\tau) =$ \textsc{GeneralizedTopOrder}($G'$).  Every edge $(u,v) \notin E(S)$ will only go forward in $\tau$, but each edge $(u,s)$, for $s \in S$, can go ``backwards" in $\tau$. By the nesting property of generalized topological orders, the amount that $(u,s)$ goes backwards -- i.e. the quantity $|\tau(X^u) - \tau(X^s)|$ -- is upper bounded by the size of the SCC in $G'$ from which $s$ was chosen: the original SCC has size $n$, but as we add vertices to $S$, the SCCs of $G' = G \setminus E(S)$ decompose and new vertices added to $S$ may belong to smaller SCCs. Define $S^* \subseteq S$ to contain all vertices $s \in S$ that were chosen in an SCC of size $\Omega(n)$. Intuitively, $S^*$ is the top-level separator chosen in $G$, before SCCs decompose into significantly smaller pieces. Every edge in $E(S^*)$ may go backwards by as much as $n$ in $\tau$, so for any path $\pi_{x,y}$ in $G$, the best we can guarantee is that $\mathcal{T}(\pi_{x,y}) \sim n \cdot |\pi_{x,y} \cap S^*|$. 

The framework of \cite{gutenberg2020decremental} tries to find a \emph{small} separator $S^*$ and then uses the trivial upper bound $|\pi_{x,y} \cap S^*| \leq |S^*|$. In fact, one can show that given \emph{any} deterministic separator procedure, the adversary can pick a sequence of updates where $|\pi_{x,y} \cap S^*| \sim |S^*|$. But now, say that $G$ is a $\sqrt{n} \times \sqrt{n}$-grid graph with bidirectional edges. It is not hard to check that  $|S^*| = \Omega(\sqrt{n})$, because every balanced separator of a grid has $\Omega(\sqrt{n})$ vertices. The framework of \cite{gutenberg2020decremental} can thus at best guarantee  $\mathcal{T}(\pi_{x,y}) \sim n \cdot |\pi_{x,y} \cap S^*| \sim n|S^*| = \Omega(n^{1.5})$, which is a $\sqrt{n}$ factor higher than it would be in a DAG, and thus leads to running time $\tilde{O}(n^{2.5})$ instead of $\tilde{O}(n^2)$.

Our algorithm uses an entirely different random separator procedure. We allow $S^*$ to be arbitrarily large, but use randomness to ensure that $|\pi_{x,y} \cap S^*|$ is nonetheless small.

\subsection{Our Improved Framework}

We now introduce our new separator procedure and then show how it can be used in a recursive algorithm that uses ATOs of worse quality (large $q$) to compute ATOs of better quality (small $q$). (By contrast, the framework of \cite{gutenberg2020decremental} could not benefit from a multi-layered algorithm because it would still hit upon the fundamental limitation outlined above.)

%As pointed out in the last section, a preliminary for obtaining a faster implementation for decremental SSSP via $\mathcal{ATO}$s is to obtain a better separator procedure. We devise such a procedure in the next paragraph and then show how it can be used to introduce a clever layering of the previous framework, resulting in a $\tilde{O}(n^2)$ algorithm for decremental SSSP. 

\paragraph{A New Separator Procedure.} Before we describe the separator procedure, let us formally define the guarantees that we obtain. In the lemma below, think of $\zeta = \Theta(\log(n))$.

\begin{restatable}{lemma}{randomLayerSep}
\label{lma:sepIntro}
There exists a procedure $\textsc{OutSeparator}(r, G, d, \zeta)$ where $G$ is a weighted graph, $r \in V$ a root vertex, and integers $d, \zeta > 0$. Then, with probability at least $1- e^{-\zeta}$, the procedure computes a tuple $(E_{Sep}, V_{Sep})$ where edges $E_{Sep} \subseteq E$, and vertices $V_{Sep} = \{ v \in V | r \leadsto_{G \setminus E_{Sep}} v\}$ such that
\begin{enumerate}
    \item for every vertex $v \in V_{Sep}$, $\mathbf{dist}_{G \setminus E_{Sep}}(r,v) \leq d$, and
    \item for every $e \in E$, we have $\mathbb{P}[e \in E_{Sep} | r \leadsto_{G \setminus E_{Sep}} \mathbf{tail}(e)] \leq \frac{\zeta}{d}w(e)$\label{prop:lowProb}.
\end{enumerate}
Otherwise, it reports $\mathbf{Fail}$. The running time of $\textsc{OutSeparator}(\cdot)$ can be bounded by $O(|E(V_{Sep})|\log n)$.
\end{restatable}

In fact, \Cref{alg:OutseparatorIntro} gives a simple implementation of procedure $\textsc{OutSeparator}(\cdot)$. Here, we pick a ball $B = B^{out}(r, X)$ in the graph $G$ from $r$ to random depth $X$, and then simply return the tuple $(E_{Sep}, V_{Sep}) = (E(B, \overline{B}), B)$ where $E(B, \overline{B})$ are the edges $(u,v)$ with $u \in B$ but $v \not\in B$. The procedure thus only differs from a standard edge separator procedure in that we choose $X$ according to the exponential distribution $\textsc{Exp}(\frac{\zeta}{d})$.

\begin{algorithm}
\caption{$\textsc{OutSeparator}(r, G, d, \zeta)$}
\label{alg:OutseparatorIntro}
Choose $X \sim \textsc{Exp}(\frac{\zeta}{d})$.\;
\lIf{$X \geq d$}{\Return \textbf{Fail}\label{lne:failSeparatorIntro}}
Compute the Ball $B = B^{out}(r, X) = \{ v \in V | \mathbf{dist}(r,v) \leq X\}$\;
\Return $(E(B, \overline{B}), B)$
\end{algorithm}

A proof of Lemma \ref{lma:sepIntro} is now straightforward. We return \textbf{Fail} in \Cref{lne:failSeparatorIntro} with probability $\mathbb{P}[X \geq d] = 1 - F_X(d) = 1 - (1 - e^{- \frac{\zeta}{d} \cdot d}) = e^{-\zeta}$ (recall from \cref{sec:prelim} that $F_X(d)$ is shorthand for $F(x, \frac{\zeta}{d})$, the cumulative distribution function of an exponential distribution with parameter $ \frac{\zeta}{d}$). Assuming no failure, we have $E_{Sep} = E(B, \overline{B})$, and it is easy to see that $V_{Sep} = \{ v \in V | r \leadsto_{G \setminus E_{Sep}} v \} = B$, so Property 1 of Lemma \ref{lma:sepIntro} holds by definition of $B$. Moreover, we can compute $B$ in the desired $O(|E(B)|\log n)$ time by using Dijkstra's algorithm by only extracting a vertex from the heap if it is at distance at most $X$. Finally, for property \ref{prop:lowProb}, note that $e \in E_{sep}$ iff $\dist(r, \mathbf{tail}(e)) \leq X < \dist(r, \mathbf{tail}(e)) + w(e)$. Thus,
\begin{align*}
    \mathbb{P}[e \in E_{Sep} | r \leadsto_{G \setminus E_{Sep}} \mathbf{tail}(e)]
    &= \mathbb{P}[X < \mathbf{dist}(r,\mathbf{tail}(e)) + w(e) | X \geq \mathbf{dist}(r,\mathbf{tail}(e))]\\
    &= \mathbb{P}[X < w(e)] = F_X(w(e)) = 1 - e^{- \frac{\zeta}{d} w(e)} \\
    &\leq 1 - \left(1 - \frac{\zeta}{d} w(e)\right) = \frac{\zeta}{d} w(e)
\end{align*}
where the second equality follows from the memory-less property of the exponential distribution, and the inequality holds because $1 + x \leq e^x$ for all $x \in \mathbb{R}$. We point out that the technique of random ball growing using the exponential distribution is not a novel contribution in itself and has be previously used in the context of low-diameter decompositions \cite{linial1993low, bartal1996probabilistic, miller2013parallel, pachocki2018approximating} which have recently also been adapted to dynamic algorithms \cite{forster2019dynamic, chechik2020dynamic}.  

\paragraph{A New Framework.} Let us now outline how to use \Cref{lma:sepIntro} to derive 
\Cref{thm:TopologicalOrderMaintenanceOverview} which is stated below again. The full details and a rigorous proof are provided in \Cref{sec:ATORealImpl}.

\ATOResultIntro*

As in \cite{gutenberg2020decremental}, we maintain a graph $G' \subseteq G$ and its generalized topological order $(\mathcal{V}, \tau)$. Whenever the diameter of an SCC $X$ in $G'$ is larger than $\frac{|X|\eta_{diam}}{n}$, we now use the separator procedure described in \Cref{lma:sepIntro} with $d = \frac{|X|\eta_{diam}}{2n}$ from some vertex $r$ in $X$ with $|B^{out}(r, d = \frac{|X|\eta_{diam}}{2n})| \leq |X|/2$. Such a vertex exists since by definition of diameter, as we can find two vertices with disjoint balls. We obtain an edge separator $E_{Sep}$ and update $G'$ by removing the edges in $E_{Sep}$. Let the union of all edge separators be denoted by $F$ and again observe that $G' = G \setminus F$. It is not hard to see that our scheme still ensures properties \ref{prop:VUpdate}-\ref{prop:ContractLittle} in \Cref{def:ATO}. We now argue that the quality improved to $\tilde{O}(n/2^i)$.

Partition $F$ into sets $F_0, F_1, \dots, F_{\lg n}$ where each $F_j$ contains all separator edges found on a graph of size $[n/2^{j+1}, n/2^j)$. We again have that every edge $(u,v) \in F_j$ has $|\tau(X^u) - \tau(X^v)| \leq n/2^{j}$. Let us establish an upper bound on the number of edges in $F_j$ on any shortest path $\pi_{s,t}$. Consider therefore any edge $e \in \pi_{s,t}$, at a stage where both endpoints of $e$ are in a SCC $X$ of size $[n/2^{j+1}, n/2^j)$ and where we compute a tuple $(V_{Sep}, E_{Sep})$ from some $r \in X$. Now, observe that if $\mathbf{tail}(e) \not\in V_{Sep}$, then $e$ can not be in $E_{Sep}$. So assume that $\mathbf{tail}(e) \in V_{Sep}$. Then, we have $e$ joining $F_j$ with probability 
\[
\mathbb{P}[e \in E_{Sep} | r \leadsto_{H \setminus E_{Sep}} \mathbf{tail}(e)] = \frac{\zeta}{2^i|X|/n} w(e) = \tilde{O}\left(\frac{2^j w(e)}{2^i}\right)
\]
according to \Cref{lma:sepIntro}, where we set $\zeta = \tilde{O}(1)$ to obtain high success probability. But if $e$ did not join $F_j$ at that stage, then it is now in a SCC of size at most $n/2^{j+1}$ (recall we chose $|B^{out}(r, d)| \leq |X|/2$, and we have $V_{Sep} \subseteq B^{out}(r, d)$). Thus, $e$ cannot join $F_j$ at any later stage. 

Now, it suffices to sum over edges on the path $\pi_{s,t}$ and indices $j$ to obtain that
\[
    \mathcal{T}'(\pi_{s,t}, \tau) \leq \sum_{e \in \pi_{s,t}} \sum_{j} \mathbb{P}[e \in F_j] \cdot n/2^{j} = \sum_{e \in \pi_{s,t}} \sum_{j} \tilde{O}\left(\frac{2^j w(e)}{2^i}\right) \cdot n/2^{j} = \tilde{O}\left(\frac{ w(\pi_{s,t}) n}{2^i}\right)
\]
giving quality $\tilde{O}\left(\frac{n}{2^i}\right)$. 

\paragraph{Efficiently Maintaining $G'$.} 
\newcommand{\asssp}{\mathcal{A}_{SSSP}}
As shown above, maintaining an $\mathcal{ATO}(G,2^i)$ requires detecting when any SCC in $G' = G \setminus E(S)$ has diameter above $2^i$. We start by showing how to do this efficiently if we are given a black-box algorithm $\asssp$ that maintains distance estimates up to depth threshold $2^i$  (i.e. if a vertex is at distance less than $2^i$ from the source vertex, there is a distance estimate with good approximation ratio). 

We use the random source scheme introduced in \cite{roditty2008improved} along with some techniques developed in \cite{chechik2016decremental, bernstein2019decremental, gutenberg2020decremental}: we choose for each SCC $X$ in $G'$ a center vertex $\textsc{Center}(X) \in X$ uniformly at random, and use $\asssp$ to maintain distances from $\textsc{Center}(X)$ to depth $2^i|X|/n \leq 2^i$. Since the largest distances between the vertex $\textsc{Center}(X)$ and any other vertex in $X$ is a $2$-approximation on the diameter of $G[X]$, this is sufficient to monitor the diameter and to trigger the separator procedure in good time.

%Further, whenever an SCC $X$ is decomposed into SCCs $X_1, X_2, \dots, X_k$ due to an adversarial edge weight increase or a new separator, we have that $\textsc{Center}(X)$ is the center of some $X_l$ and we simply remove the vertices $X \setminus X_l$ from the SSSP data structure. For the other SCCs $X_1, X_2, \dots, X_{l-1}, X_{l+1}, \dots, X_k$, we choose a new center respectively and start a new SSSP data structure on the subgraph of $G'$ induced by the SCC vertices. Intuitively, because the center is picked at random, $\textsc{Center}(X)$ usually remains in the largest SCC, so every time a vertex (or edge) is added to a new SSSP data structure, its SCC component has decreased significantly in size. In fact, Roditty and Zwick \cite{roditty2008improved} show that each vertex only participates in $O(\log n)$ SSSP data structures in expectation. 

Cast in terms of our new ATO-framework, the previous algorithm of \cite{gutenberg2020decremental} used a regular Even and Shiloach tree for the algorithm $\asssp$. We instead use a recursive structure, where $\mathcal{ATO}s$ of bad quality (large $q$) are used to build $\mathcal{ATO}s$ of better quality (small $q$). Recall that our goal is to build an $\mathcal{ATO}(G, 2^{i})$ of quality $\tilde{O}(n/2^i)$ and say that $X$ is some SCC of $G' = G \setminus E(S)$ whose diameter we are monitoring. Now, using the lower level of the recursion, we inductively assume that we can maintain an $\mathcal{ATO}(G'[X], 2^{i-1})$ of quality $\tilde{O}(n/2^{i-1})$ in time $\tilde{O}(|X|^2)$. Plugging this $\mathcal{ATO}$ into \Cref{thm:SSSPEfficientIntro} gives us an algorithm for maintaining distances up to depth $2^i$ in $X$ with total update time $\tilde{O}(|X|^2)$, Summing over all components $X$ in $G'$, we get an $\tilde{O}(n^2)$ total update time to maintain $\mathcal{ATO}(G, 2^{i})$, as desired.

%To make this scheme truly efficient, we use induction over $i$ and for each SCC $X$ whose diameter is monitored, we implemented the SSSP data structure to restricted depth using \Cref{thm:SSSPEfficientIntro} on an $\mathcal{ATO}(G'[X], 2^{i-1})$. Since this $\mathcal{ATO}$ has quality $\tilde{O}(n/2^{i-1})$ as described in \Cref{thm:TopologicalOrderMaintenanceOverview}, we have that the running time of both data structures is then $\tilde{O}(|X|^{2})$ and summing over all data structures gives total update time $\tilde{O}(n^{2})$, as desired.

We actually cheated a bit in the last calculation, because the scheme above could incur an additional logarithmic factor for computing $\mathcal{ATO}(G, 2^i)$ from all the $\mathcal{ATO}(G'[X], 2^{i-1})$, so we can only afford a sublogarithmic number of levels, which leads to an extra $n^{o(1)}$ factor in the running time. However, a careful bootstrapping argument allows us to avoid this extra term.

%reduces the total update time $\tilde{O}(n^2 \log^4 W)$, as desired. 

\subsection{A Framework for Sparse SSSP}

Similarly to \Cref{thm:TopologicalOrderMaintenanceOverview}, we can prove a similar theorem with better running time for sparse graphs.

\begin{restatable}{theorem}{ATOSparseResultIntro}[see \Cref{sec:ATORealImpl} and \Cref{subsec:bootstrapSparse}.]
\label{thm:TopologicalOrderMaintenanceOverviewSparse}
For any $0 \leq i \leq \log(Wn)$, given a decremental digraph $G=(V,E,w)$, we can maintain an $\mathcal{ATO}(G, 2^i)$ of expected quality $\tilde{O}(n/2^i)$. The algorithm runs in total expected update time $\tilde{O}(mn^{2/3})$ against a non-adaptive adversary with high probability. 
\end{restatable}

In the remaining section, let us sketch how we obtain an efficient algorithm to compute SSSP given the result from \Cref{thm:TopologicalOrderMaintenanceOverviewSparse}. To simplify the presentation let us assume for the rest of the section that the graph $G =(V,E)$ is unweighted.

\paragraph{Reducing SSSP to Hopset Maintenance.} Let us introduce the notion of a $(1+\epsilon, h)$-hopset $H$ which is a weighted decremental graph on the same vertex set as $G$ such that at any stage, for any two vertices $s,t \in V$, we have
\begin{equation}\label{eq:hopset}
    \mathbf{dist}_{G}(s,t) \leq \mathbf{dist}^h_{G \cup H}(s,t) \leq (1+\epsilon) \mathbf{dist}_{G}(s,t)    
\end{equation}
where we use the notation $\mathbf{dist}^{\ell}_{F}(s,t)$ to denote the shortest $s$ to $t$ path in the graph $F$ consisting of at most $\ell$ edges. Using a well-known rounding technique, we can then run a slightly modified ES-tree data structure from a root vertex $r$ on the graph $G \cup H$ to depth $h$ to obtain $(1+\epsilon)$-approximate distances in $G$ from $r$. The data structure only requires time $\tilde{O}((m + |E(H)|)h)$. 

In our paper, we show how to construct for every $\lg (n^{2/3} \log n) \leq i \leq \lg n$, a $(1+\epsilon, n^{2/3})$-hopset $F_i$ with $\tilde{O}(n)$ edges that satisfies equation \Cref{eq:hopset} for all pairs at distance $[2^i, 2^{i+1})$, so the total update time of the ES-tree becomes $\tilde{O}(mn^{2/3})$, as desired. For $i < \lg (n^{2/3} \log n)$, we can maintain distances up to $2^i$ in total update time $\tilde{O}(mn^{2/3})$ by simply running a classic ES-tree, without any hop-set. We obtain final distance estimates by running the above algorithm for each $i$; then, for any pair $s,t$, we find the smallest distance estimate among these $O(\log(n))$ data structures and output that as the final distance estimate.

\paragraph{Maintaining a Hopset $F_i$.} For each $\lg (n^{2/3} \log n) \leq i \leq \lg n$, we want to maintain $F_i$ as a weighted graph with the guarantee that for every vertices $s$ and $t$ at distance at $[2^{i}, 2^{i+1})$, we have a $(1+\epsilon)$-approximate shortest path in $F_i \cup G$ of hop at most $\tilde{O}(n^{2/3})$.  

To maintain such a graph $F_i$, we use $(\mathcal{V}_i, \tau_i)$ an $\mathcal{ATO}(G, \eta_{diam} = \epsilon 2^i)$ as given in \Cref{thm:TopologicalOrderMaintenanceOverviewSparse}. We first sample each vertex $v \in V$ with probability $\tilde{\Theta}(n^{1/3} / 2^i)$. We let $S$ be the set of sampled vertices and have $|S| = \tilde{O}(n^{1+1/3} / 2^i)$ with high probability. We then run from each vertex $s \in S$, with node $X^s \in \mathcal{V}_i$, an ES-tree to depth $2^i / n^{2/3}$ on the graph $G / \mathcal{V}_i$ induced by the set of nodes $Y \in \mathcal{V}_i$ such that $|\tau_i(X^s) - \tau_i(Y)| = \tilde{O}(n^{2/3})$. We then add an edge $(s,t)$ for every vertex $t$ that is in a node in the ES-tree of the vertex $s$ where we use corresponding distance estimate as an edge weight.

\paragraph{Hopset Sparsity and Running Time.} It can be shown that every edge only participates in $\tilde{O}(|S|)$ ES-trees and therefore the total running time of all ES-trees can be bound by $\tilde{O}(m |S| \cdot 2^i) = \tilde{O}(mn^{2/3})$. Further, each ES-tree from a vertex $s \in S$ has with high probability at most $\tilde{O}(n^{1/3}/2^i)$ vertices in $S$ in its tree and therefore the set $F_i$ has at most $\tilde{O}(|S|n^{1/3}/2^i) = \tilde{O}(n)$ edges. 

\paragraph{Correctness.} Finally, to show that $G \cup F_i$ contains a path of hop at most $\tilde{O}(n^{2/3})$ between any two vertices $s,t$ at distance $2^i$ let us focus on their corresponding shortest path $\pi_{s,t}$. We can partition the path into segments of length at most $2^{i-1}/n^{2/3}$, and by a standard hitting set argument, we have with high probability a vertex in $S$ in every segment. Let $s_1, s_2, \dots, s_k$ be such that each $s_j$ is a vertex in $S$ in the $j^{th}$ segment. We observe that $k \leq \frac{2^i}{2^{i-1}/n^{2/3}} = 2n^{2/3}$. Now, for every $s_j$ we either have $s_{j+1}$ in the ES-tree, in which case we have a direct edge between the vertices in $F_i$. Otherwise, some vertex $v$ on the path segment from $s_j$ to $s_{j+1}$ is in a node $Y$ such that $|\tau_i(X^s) - \tau_i(Y)| \gg n^{2/3}$. But since the quality of the $\mathcal{ATO}$ is $\tilde{O}(n/2^i)$, the total sum of topological difference of the path $\pi_{s,t}$ is only $\tilde{O}(n)$. Thus, this can occur at most $n^{1/3}$ times. Every time, we can use the shortest path in $G$ between $s_j$ and $s_{j+1}$ consisting of at most $2^{i}/n^{2/3} \leq n^{1/3}$ edges. Adding the paths from $s$ to $s_1$ and from $s_k$ to $t$, the total number of edges is at most $\tilde{O}(n^{2/3})$ as desired. Finally, we observe that when running the ES-data structure on $G \cup F_i$, we have to add an additive term of $\eta_{diam}$ since the contractions in the graph $G / \mathcal{V}_i$ might decrease distances by this additive term. But for distances in the range $[2^i, 2^{i+1})$, this term can be subsumed in the $(1+\epsilon)$ multiplicative error (after rescaling $\epsilon$ by a constant factor).

\subsection{Organization}

We recommend the reader to carefully study \Cref{sec:overview} to gain necessary intuition for our approach. In \Cref{sec:ATORealImpl}, we give an efficient reduction from maintaining an approximate topological order to depth-restricted SSSP. This section is the centerpiece of the article and its main result, \Cref{lma:reductionATOtoSSSP}, is one of our main technical contributions. 

We then show how to use \Cref{lma:reductionATOtoSSSP} to obtain a SSSP data structure for dense graphs in \Cref{sec:SSSPDense} and for sparse graphs in \Cref{sec:SSSPsparse}. Both sections follow the same structure: we first show how to obtain a depth-restricted SSSP data structure using an $\mathcal{ATO}$ and then we bootstrap the reductions to obtain the final result. 

In \Cref{sec:conclusion}, we draw a conclusion, put our results in perspective and discuss open problems.

\section{Reducing Maintenance of an Approximate Topological Order to \texorpdfstring{ $\alpha$-approximate $\delta$-restricted $\mathcal{SSSP}$}{SSSP}}
\label{sec:ATORealImpl}

In this section, we show how to obtain an $\mathcal{ATO}$ given an $\alpha$-approximate $\delta$-restricted $\mathcal{SSSP}$ data structure. We start by defining such a data structure and then give a reduction.

\begin{definition}
\label{def:SSSPGeneric}
Let $\mathcal{A}$ be a data structure that given any decremental directed weighted graph $G$, a dedicated source $r \in V$, an approximation parameter $\alpha > 1$, maintains for each vertex $v \in V$, distance estimates $\widetilde{\mathbf{dist}}(r,v)$ and $\widetilde{\mathbf{dist}}(v,r)$ such that at any stage of $G$, for every pair $(s,t) \in (\{r\} \times V) \cup (V \times \{r\})$
\begin{itemize}
    \item we have $\mathbf{dist}(s,t) \leq \widetilde{\mathbf{dist}}(s,t)$, and
    \item if $\mathbf{dist}(s,t) \leq \delta$, then $\widetilde{\mathbf{dist}}(s,t) \leq \alpha \mathbf{dist}(s,t)$.
\end{itemize}
Then, we say $\mathcal{A}$ is an $\alpha$-approximate $\delta$-restricted $\mathcal{SSSP}$ data structure with running time $T_{SSSP}(m,n,\delta,\alpha)$. We require only that $\mathcal{A}$ runs against a non-adaptive adversary. 
\end{definition}
\begin{remark}
In the rest of the article, we implicitly assume that all $\mathcal{SSSP}$ data structures have $T_{SSSP}(m,n,\delta,\alpha)$ monotonically increasing in the first two parameters.
\end{remark}

We also need another definition that makes it more convenient to work with $\mathcal{ATO}$s that only have \emph{expected} quality (see \Cref{def:ATO}). However, we require high probability bounds in our constructions and it will further be easier to work with deterministic objects. This inspires the definition of an $\mathcal{ATO}$-bundle which is a collection of $\mathcal{ATO}$s such that for each path of interest, there is at least one $\mathcal{ATO}$ in the bundle that has good quality for the path at-hand.

\begin{definition}[$\mathcal{ATO}(G, \eta_{diam}, \ell)$-bundle]
\label{def:ATObundle}
Given a decremental weighted directed graph $G = (V,E,w)$ and parameter $\eta_{diam} \geq 0$. We call $\mathcal{S} = \{ (\mathcal{V}_i, \tau_i)\}_{i \in [1, \ell]}$ an $\mathcal{ATO}(G, \eta_{diam}, \ell)$-bundle of quality $q$ if every $(\mathcal{V}_i, \tau_i)$ is an $\mathcal{ATO}(G, \eta_{diam})$ and for any two vertices $s, t \in V$, there exists an $i \in [1, \ell]$, such that the shortest path $\pi_{s,t}$ in $G$ satisfies $\mathcal{T}(\pi_{s,t}, \tau_i) \leq q \cdot w(\pi_{s,t}) + n$.
\end{definition}

Without further due, let us state and prove the main result of this section.

\begin{restatable}[$\mathcal{ATO}$-bundle from $\mathcal{SSSP}$]{theorem}{atoBundle}
\label{lma:reductionATOtoSSSP}
Given an algorithm $\mathcal{A}$ to solve $2$-approximate $\delta$-restricted $\mathcal{SSSP}$ on any graph $H$ in time $T_{SSSP}(m,n,\delta)$, and for any $c > 0$, we can maintain an $\mathcal{ATO}(G, 2\alpha\delta, 40 c \log n)$-bundle of quality $\frac{(c+2)40000 n \log^5 n}{\delta}$ in total expected update time 
\begin{equation}\label{eq:totalRunningTime}
O\left(\sum_{j=0}^{\lceil\lg \delta \rceil} \; \sum_{k=0}^{2^{j+3} c \log^2 n} T_{SSSP}(m_{j,k}, n/2^{j}, \delta, 2) + m \log^3 n\right)
\end{equation}
where $\sum_j m_{j,k} \leq 16c \cdot m \log^2 n$ for all $k$. The algorithm runs correctly with probability $1 - n^{-c}$ for any $c > 0$. 
\end{restatable}
\begin{remark}
\label{rmk:LaminarFamilyOfGraphs}
The graphs that the data structure $\mathcal{A}$ runs upon during the algorithm are vertex-induced subgraphs of $G$. The data structure $\mathcal{A}$ is further allowed to maintain distances on a larger subgraph of $G$, i.e. when $\mathcal{A}$ is applied to a graph $G[X]$, it can run instead on $G[Y]$ for any set $X \subseteq Y \subseteq V$.
\end{remark}

We point out that \Cref{rmk:LaminarFamilyOfGraphs} is only of importance at a later point at which the reader will be reminded and can safely be ignored for the rest of this sections (the reader is however invited to verify its correctness which is easy to establish).

We now describe how to obtain an efficient algorithm that obtains an $\mathcal{ATO}(G,2\alpha \delta)$ henceforth denoted by $(\mathcal{V}, \tau)$. The next sections describe how to initialize $(\mathcal{V}, \tau)$, how to maintain useful data structures to maintain the diameter, give the main algorithm and then a rigorous analysis. Finally, we obtain a $\mathcal{ATO}(G, 2\alpha\delta, 40 c \log n)$-bundle by running $40 c \log n$ independent copies of the algorithm below.

\subsection{Initializing the Algorithm} 
\label{subsec:InitATO}

As described in \Cref{sec:overview}, our goal is to maintain a graph $G'$ that is a subgraph of $G$ and satisfies that no SCC $X$ in $G'$ has weak diameter $\mathbf{diam}(X,G)$ larger than $\frac{ \delta |X|}{n}$. Throughout, we maintain the generalized topological order $(\mathcal{V}, \tau)$ on $G'$ where $\tau$ has the nesting property as described in \Cref{thm:SCCinDecrGraph}. 

To ensure the diameter constraint initially, we use the following partitioning procedure whose proof is deferred to \Cref{sec:proofPartitionFull}.

\begin{restatable}[Partitioning Procedure]{lemma}{partition}
\label{lma:partitionFull}
Given an algorithm $\mathcal{A}$ to solve $2$-approximate $\delta$-restricted $\mathcal{SSSP}$. There exists a procedure $\textsc{Partition}(G, d, \zeta)$ that takes weighted digraph $G$, a depth threshold $d \leq \delta$ and a success parameter $\zeta > 0$, and returns a set $E_{Sep} \subseteq E$ such that
\begin{enumerate}
    \item for each SCC $X$ in $G \setminus E_{Sep}$, we have for any vertices $u,v \in X$ that $\mathbf{dist}_{G \setminus E_{Sep}}(u,v) \leq d$, and \label{prop:partitionDiameter}
    \item \label{prop:partitionProb}
    for $e \in E$, we have
    $\mathbb{P}[e \in E_{Sep}] \leq \frac{240 \zeta \log^2 n}{d}w(e).$
\end{enumerate}
The algorithm runs in total expected time 
\[
O\left(\sum_{j=0}^{\lceil\lg \delta \rceil} \; \sum_{k=0}^{2^{j+1}} T_{SSSP}(m_{j,k}, n/2^{j}, \delta, 2) + m \log^2 n\right)
\]
where we have that $\sum_{k=0} m_{j,k} \leq 2m$ for every $i$. The algorithm terminates correctly with probability $1-e^{-\zeta}$ for any $c > 0$. 
\end{restatable}
\begin{remark}
During the execution, the graphs on which we use the $\mathcal{SSSP}$ structure upon have the properties as described in \Cref{rmk:LaminarFamilyOfGraphs}.
\end{remark}

\begin{algorithm}
\caption{$\textsc{Init}()$}
\label{alg:init}
Let $G'$ be initialized to $G$.\;
\For{$i = 0$ to $\lceil\lg \delta \rceil$}{
    Compute the SCCs $\mathcal{V}$ of $G'$\;
    \ForEach{SCC $X$ in $\mathcal{V}$, $|X| \leq n/2^i$}{
        $E_{Sep} \gets \textsc{Partition}(G[X], \delta/ 2^i, (c+2)\log n)$\;
        $G' \gets G' \setminus E_{Sep}$\;
    }
}
\Return $G'$\;
\end{algorithm}

Using this procedure, it is straight-forward to initialize our algorithm. The pseudo-code of the initialization procedure is given in \Cref{alg:init}. Here, we iteratively apply the partitioning procedure to SCCs of small size to decompose them further if their diameter is too large. It is not hard to establish that the graph $G'$ returned by the procedure, satisfies that every SCC $X$ in $G'$ has $\mathbf{diam}(X,G)\leq \frac{\delta |X|}{n}$. 

\subsection{Maintaining Information about SCC Diameters}
\label{subsec:maintainSSSPs}

Before we describe how to maintain $G'$ to satisfy the guarantees given above, we address the issue of maintaining information about the diameter of the current SCCs in $G'$. 

Therefore, we maintain a set $S$ of random sources throughout the algorithm, and from each $s \in S$, we run an $\alpha$-approximate $\delta$-restricted $\mathcal{SSSP}$ data structure $\mathcal{A}_s$. Initially $S = \emptyset$, and whenever there is an SCC $X$ in $\mathcal{V}$ (which is maintained by the data structure on $G'$), and we find $S \cap X = \emptyset$, we pick a vertex $s$ uniformly at random from $X$ and add it to $S$. Once added, we initialize and maintain an $\alpha$-approximate $\delta$-restricted $\mathcal{SSSP}$ data structure $\mathcal{A}_s$ on the current version of $G[X]$. That is, even if $X$ does not form an SCC at later stages, the data structure is run until the rest of the algorithm on the graph $G[X]$. This ensures that once the algorithm is invoked, all edge updates are determined by the adversary formulating updates to $G$. Since we assume that the adversary is non-adaptive, we have that the $\mathcal{SSSP}$ data structure only has do deal with updates from a non-adaptive adversary\footnote{If we would instead remove vertices from the data structure, we would do so based on the information gathered from the data structure. Thus, the data structure would be required to work against an adaptive adversary. A similar problem arises when running on $G'$.}.

We point out that since we maintain $G'$ to be a decremental graph, we have that $\mathcal{V}$ forms a refinement of previous versions at any stage i.e. the SCC sets only decompose over time in $G'$. Therefore, we can never have multiple center vertices in the same SCC $X \in \mathcal{V}$. For convenience, we let for each $X \in \mathcal{V}$, the vertex $\{s\} = X \cap S$ be denoted by $\textsc{Center}(X)$. By the above argument, this function is well-defined.

\subsection{Maintaining \texorpdfstring{$G'$}{G'}}
\label{subsec:maintainGPrime}

Let us now describe the main procedure of our algorithm: the part that efficiently handles violations of the diameter constraint by finding new separators. The implementation of this procedure is given by \Cref{alg:diameterVio}. Let us now provide some intuition and detail as to how the algorithm works.

\begin{algorithm}
\caption{$\textsc{ResolveDiameterViolations}()$}
\label{alg:diameterVio}
\While(\label{lne:loop}){there exists an $X \in \mathcal{V}$, where $\mathcal{A}_{\textsc{Center}(X)}$ has a distance estimate $\widetilde{\mathbf{dist}}(\textsc{Center}(X), t)$ or $\widetilde{\mathbf{dist}}(t, \textsc{Center}(X))$ exceeding $\frac{\delta|X|}{n}$ for some vertex $t \in X$}{
    \tcc{Find separator sets that decompose $X$.}
    \If(\label{lne:ifFar}){$\widetilde{\mathbf{dist}}(t,\textsc{Center}(X)) > \frac{|X|\delta}{n}$}{
        $(E_{Sep}, C) \gets \textsc{OutSeparator}(t, G'[X], \frac{|X|\delta}{2n}, (c+2)\log n)$ \label{lne:DelSep}
    }
    \Else{
        $(E_{Sep}, C) \gets \textsc{OutSeparator}(t, \overleftarrow{G'[X]}, \frac{|X|\delta}{2n}, (c+2)\log n)$ \label{lne:DelSepRev}
    }
    $E'_{Sep} \gets \textsc{Partition}(G'[C], \frac{|X|\delta}{4n}, (c+2)\log n))$\label{lne:partitionXC}\;
    \tcc{Update $G'$, $\mathcal{V}$ and $\tau$ to reflect the changes.}
    $G' \gets G' \setminus (E_{Sep} \cup E'_{Sep})$\label{lne:updateGraph}\;
    \textbf{Wait Until} the generalized topological order $(\mathcal{V}, \tau)$ of $G'$ was updated, each SCC $Z$ in $G'$ has a center $\textsc{Center}(Z)$, and all data structures $\mathcal{A}_{s}$ are updated.\label{lne:WaitUntil}
}
\end{algorithm}

The algorithm runs a while-loop starting in \Cref{lne:loop} that checks whether there exists a SCC $X \in \mathcal{V}$, such that the $\alpha$-approximate $\delta$-restricted SSSP data structure $\mathcal{A}_{\textsc{Center}(X)}$ has one of its distance estimates $\widetilde{\mathbf{dist}}(\textsc{Center}(X), t)$ (or $\widetilde{\mathbf{dist}}(t, \textsc{Center}(X))$) exceeding $\frac{\delta|X|}{n}$ for some vertex $t$ in the same SCC $X$ in $G'$. The goal of the while-loop iteration, is then to find a separator $E_{Sep}$ between $\textsc{Center}(X)$ and $t$ and to delete the edges from $G'$.

Let us describe a loop-iteration where some distance estimate $\widetilde{\mathbf{dist}}(\textsc{Center}(X), t)$ was found that exceeded $\frac{\delta|X|}{n}$ and where $t \in X$ (the case where we have a distance estimate $\widetilde{\mathbf{dist}}(\textsc{Center}(X), t)$ exceed the threshold value is analogous and therefore omitted). In this case, we find a separator $E_{Sep}$ that separates vertices in $C$ (where $t \in C$) from vertices in $X \setminus C$ (where $\textsc{Center}(X) \in X \setminus C$) in $G'$. Further, we invoke the procedure $\textsc{Partition}(G'[C], \frac{\delta|X|}{4n}, \zeta)$ on $C$ and obtain a separator $E'_{Sep}$ in $G'$ such that each SCC in $G'[C] \setminus E'_{Sep}$ has small diameter. We point out that while the first separator procedure is necessary to separate the vertices $\textsc{Center}(X)$ and $t$ in $G'$, the partitioning procedure is run for technical reasons only since we cannot ensure an efficient implementation without this step.

Finally, we wait until the data structures that maintain the generalized topological order and the distance estimates from random sources are updated before we continue with the next iteration. On termination of the while-loop, we have that all distance estimates between centers and vertices in their SCC (with regard to $G'$) are small (with regard to $G$).

\subsection{Analysis}

We establish \Cref{lma:reductionATOtoSSSP} by establishing four lemmas establishing for $(\mathcal{V}, \tau)$ correctness (\Cref{lma:correctness}), running time (\Cref{lma:runningTime}) and success probability (\Cref{lma:successRate}) and finally establishing that $c \log n$ independent copies of $(\mathcal{V}, \tau)$ form an $\mathcal{ATO}(G,2\alpha\delta)$-bundle with the guarantees given in \Cref{lma:reductionATOtoSSSP}, as required.

\begin{lemma}[Correctness]
\label{lma:correctness}
Given that no procedure returns $\textbf{Fail}$, we have that the algorithm maintains $(\mathcal{V}, \tau)$ to be an $\mathcal{ATO}(G,2\alpha \delta)$ of expected quality $\frac{(c+2)20000 n \log^5 n}{\delta}$.
\end{lemma}

Let us first prove that the diameter of SCCs in $G'$ remains small.

\begin{claim}
\label{lma:smallDiam}
After invoking \Cref{alg:diameterVio}, we have that each set $X \in \mathcal{V}$ satisfies
\[
\mathbf{diam}(X, G) \leq \frac{2\alpha\delta|X|}{n}
\] 
\end{claim}
\begin{proof}
First, recall that when the while-loop in \Cref{lne:loop} terminates, we have that every $X \in \mathcal{V}$ has that no distance estimate $\widetilde{\mathbf{dist}}(\textsc{Center}(X), t)$ or $\widetilde{\mathbf{dist}}(t, \textsc{Center}(X))$ exceeds $\frac{\delta|X|}{n}$ for any $t \in X$. 

Next, observe that the algorithm maintains the following invariant on the while-loop in \Cref{lne:loop}: every $X \in \mathcal{V}$ contains exactly one center is only marked in the data structure $\mathcal{E}_{\textsc{Center}(X)}$. This follows by resampling centers in SCCs $X$ that do not have a center yet and the by \Cref{lne:WaitUntil} which ensures that at the end of each while-loop iteration, there is time to resample.

Combined, this implies that on termination of the while-loop, for every $x, y \in X$, in any $X \in \mathcal{V}$, we have
\begin{align}
\begin{split}
\label{eq:upperBound}
    \mathbf{dist}_G(x,y) &\leq \mathbf{dist}_G(x,\textsc{Center}(X)) + \mathbf{dist}_G(\textsc{Center}(X),y) \\
    &\leq \widetilde{\mathbf{dist}}(x,\textsc{Center}(X)) + \widetilde{\mathbf{dist}}(\textsc{Center}(X),y) \\
     & \leq \frac{2\alpha\delta|X|}{n}
\end{split}
\end{align}
where we used the triangle inequality, \Cref{def:SSSPGeneric} and the fact that $\mathcal{A}_{\textsc{Center}}$ maintains distances with regard to a vertex-induced subgraph of $G$ (adding edges can only decrease distances, thus distances in $G$ are smaller than in $G[Y] \subseteq G$ for any $Y$). 
\end{proof}

Let us now bound the quality of the approximate topological order $(\mathcal{V}, \tau)$, i.e. upper bound for any $s$-to-$t$ path $\pi_{s,t}$ the amount $\mathcal{T}(\pi_{s,t}, \tau)$. As in the overview section, we focus on the "negative" terms in $\mathcal{T}(\pi_{s,t}, \tau)$, which are captured by 
\begin{equation}
        \mathcal{T}'(\pi_{s,t}, \tau) \stackrel{\text{def}}{=} \sum_{(u,v) \in \pi_{s,t}} \min\{0, \tau(X^v) - \tau(X^u)\}
\end{equation}
which is the definition of $\mathcal{T}'$ already given in equation \ref{eq:redefT}. It is not hard to see that $\mathcal{T}(\pi_{s,t} , \tau) = 2\mathcal{T}'(\pi_{s,t} , \tau) + |\tau(X^s) - \tau(X^t)| \leq 2\mathcal{T}'(\pi_{s,t} , \tau) + n$. It, thus, only remains to establish the following lemma.

\begin{claim}
\label{lma:ProbOfEi}
At any stage of $G$, for any path $\pi_{s,t}$ in $G$, we have 
\[
    \mathbb{E}[\mathcal{T}'(\pi_{s,t}, \tau)] \leq \frac{(c+2)10000 n \log^5 n}{\delta} w_G(\pi) 
\]
throughout the course of the algorithm.
\end{claim}

Before, we provide a proof, let us state the following lemma which has been shown in slightly different forms in various papers before and whose proof is therefore delayed to \Cref{sec:EStreeprobProof}.

\begin{restatable}[c.f. also \cite{chechik2016decremental}, Lemma 13; \cite{bernstein2019decremental}, Lemma 7.1]{lemma}{participation}
\label{lma:EStreeprob}
Each vertex $v \in V$ participates in $C$ in at most $2 \lceil\lg \delta\rceil$ while-loop iterations of \Cref{alg:diameterVio} during the entire course of the algorithm in expectation. Further, in expectation, the SCC in $G'$ that $v$ is contained in, halves every second time that $v$ participates in $C$. 
\end{restatable}

\begin{proof}[Proof of \Cref{lma:ProbOfEi}]
We proof this lemma for edges $(u,v) \in E$. Then, the result follows straight-forwardly by summing over the path edges. Let us start by observing that we have $\mathcal{T}'((u,v), \tau) \neq 0$ if and only if $X^v$ strictly precedes $X^u$ in $\tau$ (where $X^z$ denotes the set in $\mathcal{V}$ that contains vertex $z \in V$). But since $(\mathcal{V}, \tau)$ forms a generalized topological order of $G'$, we have that $(u,v)$ cannot be contained in $E(G')$. 

However, we only remove edges from $E(G')$ in  \Cref{lne:updateGraph} of our algorithm, after being added to $E_{Sep}$ in  \Cref{lne:DelSep} or \ref{lne:DelSepRev}, or to $E'_{Sep}$ in  \Cref{lne:partitionXC}. Having $(u,v) \in E_{Sep}$ occurs by \Cref{lma:sepIntro} only if $(u,v)$ is contained in $G'[X]$ and if at least one of the endpoints is in $C$ (depending on whether the separator is computed on $G'[X]$ or $\overleftarrow{G'[X]}$ it is $u$ or $v$). In this case, the probability that $(u,v)$ is added to $E_{Sep}$ is at most $\frac{(c+2)\log n 2n}{|X|\delta} w_G(u,v)$, again by \Cref{lma:sepIntro}. 

However, if $(u,v)$ is not added to $E_{Sep}$ (and not already removed from $G'$) then it is completely contained in $G'[C]$. Thus, by \Cref{lma:partitionFull} it is sampled into $E'_{Sep}$ with probability at most
$\frac{(c+2)240 \log^4(n) \cdot 4n}{|X|\delta}w_G(u,v) \leq \frac{(c+2)960 n \log^4 n}{|X|\delta}w_G(u,v)$.

Observe that if $(u,v)$ is sampled into either $E_{Sep}$ or $E'_{Sep}$, then since it was contained in $X$ and by the nesting property of $\tau$ which is guaranteed by \Cref{thm:SCCinDecrGraph}, we have that during the rest of the algorithm, we have $|\tau(X^u) - (X^v)| < |X|$ where $X^u$ (resp. $X^v$) denotes the set in $\mathcal{V}$ that contains $u$ (resp. $v$).

Thus, a while-loop iteration where $u$ or $v$ participate in $C$ adds to $\mathbb{E}[\mathcal{T}'(\pi_{s,t}, \tau)]$ at most
\[
|X| \cdot \frac{(c+2)1000 n \log^4 n}{|X|\delta}w_G(u,v) = \frac{(c+2)1000 n \log^4 n}{\delta}w_G(u,v).
\]
Since by \Cref{lma:EStreeprob} each vertex only occurs during $2 \lg n$ while-loop iterations in $C$, we can establish the final bound.
\end{proof}

Combining the fact that $(\mathcal{V}, \tau)$ is a $\textsc{GeneralizedTopologicalOrder}(G')$ at all stages and $G' \subseteq G$ where $\tau$ has the nesting property, combined with \Cref{lma:smallDiam} and \cref{lma:ProbOfEi}, we derive \Cref{lma:correctness}.

\begin{lemma}[Running Time]
\label{lma:runningTime}
The algorithm to maintain $(\mathcal{V}, \tau)$ requires at most expected time
\[
O\left(\sum_{j=0}^{\lceil\lg \delta \rceil} \; \sum_{k=0}^{2^{j+3} \lceil \lg \delta\rceil} T_{SSSP}(m_{j,k}, n/2^{j}, \delta, 2) + m \log^2 n\right)
\]
where $\sum_j m_{j,k} \leq 16m \lceil \lg \delta\rceil$.
\end{lemma}
\begin{proof}
Again, our proof crucially relies on the following lemma.

\participation*

We first observe that the initialization procedure described in \Cref{subsec:InitATO} initializes $G'$ in $O(m)$ time and then runs $O(\log n)$ iterations where in each iteration it invokes the procedure $\textsc{Partition}(\cdot)$ on a set of disjoint subgraphs of $G$ to update $G'$. By \Cref{lma:partitionFull}, we can implement all of these calls in time 
\[
O\left(\log n \left(\sum_{j=0}^{\lceil\lg \delta \rceil} \; \sum_{k=0}^{2^{j+1}} T_{SSSP}(m_{j,k}, n/2^{j}, \delta, 2) + m \log n\right)\right).
\]
The latter term in this expression subsumes the time spend on updating $G'$ once a separator is returned.

Next, let us bound the time spend on maintaining the $\mathcal{SSSP}$ data structures as described in \Cref{subsec:maintainSSSPs}. It is here that we use \Cref{lma:EStreeprob}: we have that initially each vertex (and edges) is in exactly one data structure. Further, every second time a vertex $v$ participates in $C$ as computed in \Cref{lne:DelSep} or \Cref{lne:DelSepRev}, the SCC it is contained in in $G'$ is halved in size (i.e. in the number of vertices). Since new $\mathcal{SSSP}$ data structures are initialized on the new SCCs that are contained in the $C$ set, we have that each vertex $v$, in expectation, only participates $2$ times in an $\mathcal{SSSP}$ structure with running time $T_{SSSP}(m_{j,k}, n/2^{j}, \delta, 2) + m \log n$ for any $j$. Since each edge is incident to only two vertices, we have a similar argument on edges and can therefore bound the total amount of time spend on $\mathcal{SSSP}$ data structures by
\[
O\left(\sum_{j=0}^{\lceil\lg \delta \rceil} \; \sum_{k=0}^{2^{j+1}} T_{SSSP}(m_{j,k}, n/2^{j}, \delta, 2)\right).
\]
where $\sum_j m_{j,k} \leq 4m$.

Finally, let us bound the time spend in calls to \Cref{alg:diameterVio}. We observe that each while-loop iteration takes time 
\[
O\left(\sum_{j=0}^{\lceil\lg \delta \rceil} \; \sum_{k=0}^{2^{j+1}} T_{SSSP}(m'_{j,k}, n'/2^{j}, \delta, 2) + m' \log n\right)
\]
where $\sum_j m'_{j,k} \leq 4m'$ for $m' = |E_G(C)|$ and $n' = |V_G(C)|$. This follows since the $\textsc{OutSeparator}(\cdot)$ procedure runs in time almost-linear in the number of edges incident to $C$ and afterwards the call of the procedure $\textsc{Partition}(\cdot)$ which dominates the costs of the procedure is only on the graph $G'$ induced by the vertices in $C$. Thus, this insight follows straight-forwardly from \Cref{lma:sepIntro} and \Cref{lma:partitionFull} and the insight that the cost of the remaining operations is subsumed in the bounds.

Finally, we again use \Cref{lma:EStreeprob} which gives that summing over all while-loop iterations is at cost at most 
\[
O\left(\sum_{j=0}^{\lceil\lg \delta \rceil} \; \sum_{k=0}^{2^{j+2} \lceil \lg \delta\rceil} T_{SSSP}(m_{j,k}, n/2^{j}, \delta, 2) + m \log^2 n\right)
\]
where $\sum_j m_{j,k} \leq 8m \lceil \lg \delta\rceil$. Combining the parts of the algorithm, we thus get the total bound.
\end{proof}

\begin{lemma}[Success Probability]
\label{lma:successRate}
The algorithm reports $\textbf{Fail}$ with probability at most $2n^{-c-1}$. 
\end{lemma}
\begin{proof}
We point out that we can only get a $\textbf{Fail}$ due to procedures $\textsc{OutSeparator}(\cdot)$ and $\textsc{Partition}(\cdot)$. 

Since each separator found in the while-loop in \Cref{lne:loop} refines $\mathcal{V}$, we can bound the number of while-loop iterations in the course of the algorithm by $n-1$. Thus, we make at most $n-1$ calls to procedures $\textsc{OutSeparator}(\cdot)$ and $\textsc{Partition}(\cdot)$. Each of the former calls returns $\textbf{Fail}$ with probability at most $n^{-(c+2)}$ and each of the latter with probability at most $n^{-(c+2)}$. 

Taking a union bound over all events, the lemma follows.
\end{proof}

Finally, let us put everything together and prove our main theorem.

\atoBundle*

\begin{proof}
We maintain a collection of $40 c \log n$ independent $\mathcal{ATO}(G, 2\alpha\delta)$ instances 
\[
(\mathcal{V}_1, \tau_1), (\mathcal{V}_2, \tau_2), \dots ,(\mathcal{V}_{40c \log n}, \tau_{40c \log n})
\]
as described earlier in this section and let $\mathcal{S}$ denote the collection of these instances.

The total running time to maintain these $\mathcal{ATO}(G, 2\alpha\delta)$'s is clearly bounded by the term given in equation \ref{eq:totalRunningTime} by \Cref{lma:runningTime}.

Now, since by \Cref{lma:correctness}, each $\mathcal{ATO}(G, 2\alpha\delta)$ has expected quality $q = \frac{(c+2)20000 n \log^5 n}{\delta}$, we have by Markov's inequality and a simple Chernoff bound, that for each shortest path $\pi_{s,t}$ in $G$ at some stage $t$, we have that there exists an $i$, such that $\mathcal{T}(\pi_{s,t}, \tau_i) \leq 2q$ with probability at least $1 - e^{- 40 c \log n/8} = 1 - n^{-5 c}$. Since $c > 1$, we have that the probability that any shortest-path at any stage fails, is at most $1-n^{-c}/2$ by union bounding over at most $n^2$ stages and at most $n^2$ shortest-paths, for $n$ large enough. Moreover, the total probability that any instance returns $\textbf{Fail}$ is at most $n^{-c}/2$ by \Cref{lma:successRate} and a union bound over the instances. Thus, we have established that with probability at least $1-n^{-c}$, $\mathcal{S}$ forms an $\mathcal{ATO}(G,2\alpha \delta)$-bundle of quality $2q$ as defined in \Cref{def:ATObundle}.
\end{proof}

\section{A SSSP Algorithm for Dense Graphs}
\label{sec:SSSPDense}

We now give a proof of  \Cref{thm:TopologicalOrderMaintenanceOverview} which implies our main result,  \Cref{thm:ContributionSSSPResult}, as a corollary. Our proof is in two steps: we first show how to implement an $\alpha$-approximate $\delta$-restricted $\mathcal{SSSP}$ as described in \Cref{def:SSSPGeneric} given access to approximate topological orders. We then show how to bootstrap the reductions to maintains different $\mathcal{SSSP}$ data structures to cover all depths.

\subsection{\texorpdfstring{$\alpha$-approximate $\delta$-restricted $\mathcal{SSSP}$}{SSSP} via Maintaining an Approximate Topological Order}
\label{subsec:alphaDeltaSSSPReduction}

The main objective of this section is to prove the following theorem which gives a reduction from $(1+\epsilon)$-approximate $\delta$-restricted $\mathcal{SSSP}$ to approximate topological orders. In this theorem, we only assume access to an $\mathcal{ATO}(G, \eta_{diam})$ denoted by $(\mathcal{V}, \tau)$ where we assess the quality individually for each path. If the quality for a certain tuple is below a threshold $q$, we show how to exploit the approximate topological order to maintain the distance estimate for the tuple efficiently, otherwise we provide no guarantees.

\begin{restatable}{theorem}{ssspSimple}
\label{thm:SSSPEfficient}
Given $G=(V,E,w)$, a decremental weighted digraph, a source $r \in V$, a depth threshold $\delta > 0$, a quality parameter $q$, an approximation parameter $\epsilon > 0$, and access to $(\mathcal{V}, \tau)$ an $\mathcal{ATO}(G, \eta_{diam})$.

Then, there exists a deterministic data structure that maintains a distance estimate $\widetilde{\mathbf{dist}}(r,v)$ for every vertex $v \in V$ such that at each stage of $G$, $\mathbf{dist}_G(r,v) \leq \widetilde{\mathbf{dist}}(r,v)$ and if $\mathbf{dist}_G(r,v) \leq \delta$ and $\mathcal{T}(\pi_{r,v},\tau) \leq q \cdot \delta + n$, then 
\[
\widetilde{\mathbf{dist}}(r,v) \leq \mathbf{dist}_G(r,v) + \eta_{diam} + \epsilon\delta.
\]
The total time required by this structure is 
\[
O(n\delta q \log n/\epsilon + n^2 \log n)
\]
\end{restatable}
\begin{restatable}{remark}{ssspRemark}
\label{rmk:ssspRemark}
Technically, we require the approximate topological order $(\mathcal{V}, \tau)$ to encode changes efficiently and pass them the SSSP data structure. Since the SSSP data structure is updated only through delete operations, we require, that with each edge update, the data structure receives changes to $(\mathcal{V}, \tau)$ since the last stage. More precisely, we require that the user passes a set of pointers to each set $Y$ that occurred in $\mathcal{V}$ at the previous stage (denoted $\mathcal{V}^{OLD}$), but did not occur in $\mathcal{V}$ at the current stage (denoted $\mathcal{V}^{NEW}$), i.e. each $Y \in \mathcal{V}^{OLD} \setminus \mathcal{V}^{NEW}$. Additionally, we require with each such $Y$ that was split into subsets $Y_1, Y_2, \dots, Y_k \in \mathcal{V}^{NEW}$ that form a partition of $Y$, pointers to each new element $Y_i$. We further require worst-case constant query time of $\tau$, and each element $Y \in \mathcal{V}$ (for any version) can be queried for its size in constant time and returns its vertex set in time $O(|Y|)$. For the rest of the paper, this detail will be concealed in order to improve readability.
\end{restatable}

Before we show how to implement such a data structure, let us emphasize that the above theorem directly implies \Cref{thm:SSSPEfficientIntro} that we introduced in the overview. It can further also be used to derive the following corollary which is at the heart of our proof in the next section. Its proof is rather straight-forward and can be found in 
\Cref{sec:ssspfromATOfullProofCor}.

\begin{restatable}{corollary}{corSSSPrestrictedFromATO}
\label{cor:SSSPtoATO}
Given $G=(V,E,w)$, a decremental weighted digraph, a source $r \in V$, a depth threshold $\delta > 0$, an approximation parameter $\epsilon > 0$, and access to a collection $\mathcal{S} = \{ S_i \}_{1 \leq i \leq \mu}$ for $\mu = \lfloor\lg \delta\rfloor-1$ where each $\mathcal{S}_i$ forms an $\mathcal{ATO}(G, 2^i, 40 c\log n)$-bundle of quality $q_i$. Then, there exists an implementation for $(1+\epsilon)$-approximate $\delta$-restricted $\mathcal{SSSP}$ where $T_{SSSP}(n,m,\delta, \epsilon) = O(n
(\max_{1 \leq i \leq \mu} \{\frac{\delta q_i}{2^i}\} + n) \log^3 n / \epsilon
^2)$.
\end{restatable}

Let us now describe the implementation of a data structure $\mathcal{E}_r$ that stipulates the guarantees given in \Cref{thm:SSSPEfficient}. 
Since the proof that this is indeed a valid implementation of \Cref{thm:SSSPEfficient} is quite similar to the proof sketch we give in \Cref{sec:overview}, we refer the reader to \Cref{sec:proofSSSP}.

\paragraph{Initialization.} Throughout the algorithm, we define $\delta_{max} = \lceil (1+\epsilon)\delta + \epsilon n/q \rceil$ and define the complete graph $\mathcal{H} = (\mathcal{V}, \mathcal{V}^2, w)$\footnote{Here, we are again slightly abusing notation by referring to $\mathcal{V}$ as partition and node set, however, context and the fact that this implicitly refers to a one-to-one correspondence between partition sets and nodes ensures that no ambiguity arises.} with weight function 
\[
w(X,Y) = \inf\{ w(x,y) | (x,y) \in E(X,Y)\}
\]
for $X, Y \in \mathcal{V}$. We use the convention that the infimum of the empty set is $\infty$. We use $\mathcal{H}$ to avoid dealing explicitly with $G/ \mathcal{V}$ which is a multi-graph, instead in $\mathcal{H}$ we use the same node set with simple edges as the infimum over weights of the multi-edges (even if there is no such edge). 

We use a standard min-heap data structure\footnote{See for example \cite{cormen2009introduction}.} $Q_{X,Y}$ over the set $E(X,Y)$ for each ordered pair $(X,Y)$ to maintain the weight $w(X,Y)$. We henceforth denote by $Q_{X,Y}.\textsc{MinValue}$ the value $w(X,Y)$ and by $Q_{X,Y}.\textsc{MinElem}$ a corresponding edge $(x,y)$ with $x \in X, y \in Y$, and use the convention of denoting the node in $\mathcal{V}$ that contains vertex $x \in V$ by $X^x$. We initialize the data structure $\mathcal{E}_r$ by constructing $\mathcal{H}$ and by running Dijkstra's algorithm\footnote{See \cite{cormen2009introduction} for an efficient implementation.} from $X^r$ on $\mathcal{H}$. We then initialize a distance estimate $\widetilde{\mathbf{dist}}(X^r, Y)$ for each $Y \in \mathcal{V}$ to $\mathbf{dist}_{\mathcal{H}}(X^r,Y)$. If we have $\widetilde{\mathbf{dist}}(X^r,Y) > \delta_{max}$ at any point in the algorithm, we set it to $\infty$. Further, we also maintain the distance estimates $\widetilde{\mathbf{dist}}(r, u)$ for each $u \in V$ equal to $\widetilde{\mathbf{dist}}(X^r, X^u) + \eta_{diam}$, i.e. every time we increase $\widetilde{\mathbf{dist}}(X^r, X^u)$, we also increase $u$'s distance estimate\footnote{We will show that $\widetilde{\mathbf{dist}}(X^r, X^u)$ is a monotonically increasing value over time.}. This allows us to henceforth focus on the distance estimates of nodes which is easier to describe. We also store the corresponding shortest-path tree $T$ truncated at distance $\delta_{max}$ that serves as a certificate of the distance estimates.

Finally, we partition for each node $X \in \mathcal{V}$, the in-neighbors set in $\mathcal{H}$ of $X$ into different buckets based on their $\tau$-distance: for each $X$, we initialize bucket $B_{-1}(X) = \{X\}$ and for $0 \leq j \leq \lg n$ we initialize the bucket $B_j(X)$ to 
\[ 
\{ 2^{j} \leq \chi(X,Y, \tau) < 2^{j+1} | Y \in N_{\mathcal{H}}^{in}(X), X \neq Y\}
\]
where we define 
\begin{equation} \label{eq:chi}
 \chi(X ,Y, \tau) \stackrel{\text{def}}{=}          \begin{cases}
                    \tau(Y) - (\tau(X) + |X| - 1) & \text{if } \tau(X) < \tau(Y) \\
                    \chi(Y, X, \tau) & \text{otherwise}
                 \end{cases}
\end{equation}
that is $\chi(\cdot)$ is similar to $\mathcal{T}(\cdot)$ (in fact $ \chi(X ,Y, \tau) \leq \mathcal{T}(X, Y, \tau)$), however, as $\tau$ maps nodes $X$ and $Y$ to disjoint intervals, $\mathcal{T}(\cdot)$ measures the distance between the starting points of the intervals, while $\chi(\cdot)$ measures the distance between the intervals (i.e. the closest endpoints of the intervals).

At any stage, we let $B_{\leq j}(X) = \bigcup_{j' \leq j} B_{j'}(X)$. We store each set $B_j(X)$ explicitly as a linked list, store for each $Y \in B_j(X)$ a pointer to the bucket, and maintain the buckets to partition the in-neighbors of each $X$. 

\paragraph{Handling Edge Deletions.} The edge deletion procedure takes two parameters: the edge to be deleted $(u,v)$ and a collection of tuples $U$ that encode refinements of $\mathcal{V}$ during this stage. To handle the update, we initialize a min-heap $Q = \emptyset$ that keeps track of the nodes in $\mathcal{H}$, that cannot be reached from $X^r$ in the truncated shortest-path tree $T$ (i.e. whose certificate for the current distance estimate was compromised).

We start our update procedure by processing updates to $(\mathcal{V}, \tau)$ (check the remark of \Cref{thm:SSSPEfficient} for a description of these updates encoded by $U$). For any node $X \in \mathcal{V}^{OLD} \setminus \mathcal{V}^{NEW}$, that was split into subsets $X_1, X_2, \dots, X_k \in V^{NEW}$ (i.e. for every tuple $(X, X_1, X_2, \dots, X_k) \in U$), we query for each $X_i$, its size. Then, we let the largest node $X_i$ \emph{inherit} the original node $X$ (that is the nodes are equal in our data structure at this stage although the partition sets are not), and create a new node $X_{i'}$ in $\mathcal{H}$ for other $i' \neq i$, and new heap structures $Q_{X_{i'}, Y}$ for every $Y \in \mathcal{V}$. Then for each $X_{i'}$, $i' \neq i$, we scan each edge $(x,y)$ in $E(X_{i'}, V \setminus X_{i'})$, remove it from the heap $Q_{X_i, X^y}$ and add it to the new heap $Q_{X_{i'}, X^y}$. We also initialize the distance estimates for each $X_{i'}$, $\widetilde{\mathbf{dist}}(X^r, X_{i'})$ to take the value $\widetilde{\mathbf{dist}}(X^r, X)$, also for $X_i$. We then find the edge $(w,x)$ in $T$ where $X = X^x$. Clearly, we now have that $X_{i'} = X^x$ for some ${i'}$ and we connect $X_{i'}$ in the tree $T$ since this edge is now a certificate for $X_{i'}$'s distance estimate. The rest of the nodes, i.e. the nodes $X_1, X_2, \dots, X_{{i'}-1}, X_{{i'}+1}, \dots, X_k$, we add to $Q$ since we do not have a certificate for them yet.

Finally, when \emph{all} the node splits for the current stage where processed, we update the buckets $B_j(Y)$ for all $j$ and $Y$ to \emph{almost} stipulate the initialization rules. We point out that all edges that have to be assigned to a different bucket have to be incident to $X_1, X_2, \dots, X_k$ by property \ref{prop:tauUpdate}. Then, for each $X_{i'}$ (also $X_i$), we compute $j$ to be the largest integer such that some number in $|X_{i'}|, |X_{i'}|+1, \dots, |X| - 1$ is divisible by $2^j$. Then, we update all nodes in $B_{\leq j+1}(X_{i'})$ by scanning and reassigning them, and similarly reassign $X_{i'}$ to a new bucket for each $Y$ where $X_{i'} \in B_{\leq j}(Y)$. 

Finally, when all node splits are processed, the node set of $\mathcal{H}$ reflects the current $\mathcal{V}$, and we can delete the edge $(u,v)$ from $\mathcal{H}$ by deleting it from the heap it is contained in. If $(u,v)$ was equal to $Q_{X^u, X^v}.\textsc{MinElem}$, and $(X^u, X^v) \in T$, we delete it from $T$ and insert $X^v$ into $Q$.

We then rebuild our certificate $T$: we take the node $Y$ from $Q$ with the smallest distance estimate $\widetilde{\mathbf{dist}}(X^r, Y)$ until $Q$ is empty or the smallest distance estimate $\infty$. Now, let $j$ be the largest integer such that $ \widetilde{\mathbf{dist}}(X^r, Y)$ is divisible by $\lceil 2^{j} \cdot \frac{\epsilon}{q} \rceil$. We then check if there exists a node $X \in B_{\leq j}(Y)$ such that
\[
    \widetilde{\mathbf{dist}}(X^r, X) + Q_{X,Y}.\textsc{MinValue} \leq \widetilde{\mathbf{dist}}(X^r, Y).
\]
In this case, the edge $(x,y) = Q_{X,Y}.\textsc{MinElem}$ serves as a certificate that the distance from $X^r$ to $Y$ is at most $\widetilde{\mathbf{dist}}(X^r, Y)$ and therefore we add $(x,y)$ to $T$. If there exists no such vertex $X$, then we increase the value $\widetilde{\mathbf{dist}}(X^r, Y)$ by one or to $\infty$ if it is currently at least $\delta_{max}$ and reinsert $Y$ and children $Z_1, Z_2, \dots, Z_k$ of $Y$ in $T$ into $Q$ (after deleting the edges $(Y, Z_i)$ from $T$). This completes the description of the algorithm. Again, we refer the reader interested in the proof of \Cref{thm:SSSPEfficient} to \Cref{sec:proofSSSP}.

\subsection{Bootstrapping an Algorithm for Unrestricted Depth}
\label{subsec:bootstrapDense}

Next, let us prove the following theorem which show is a detailed version of \Cref{thm:TopologicalOrderMaintenanceOverview}. Combined with \Cref{cor:SSSPtoATO} (where we set the depth threshold parameter $\delta$ to $Wn$), this immediately implies our main result, \Cref{thm:ContributionSSSPResult}.

\begin{theorem}
\label{thm:TopologicalOrderMaintenanceBootstrap}
For any $0 \leq i \leq \lg(Wn)$, given a decremental digraph $G=(V,E,w)$, we can maintain a hierarchy $\mathcal{S} =\{\mathcal{S}_i\}_i$ where each $\mathcal{S}_i$ is a $\mathcal{ATO}(G, 2^i, 40 c \log n)$-bundle of expected quality $\tilde{O}(n/2^i)$. The algorithm runs in total expected update time $O(c^5 n^2 \log^{17} n \lg^{5}(Wn))$ against a non-adaptive adversary and is correct with probability at least $1-n
^{-c+2}$ for any failure probability parameter $c \geq 2$.
\end{theorem}
\begin{proof}
In order to prove our theorem formally, we need to fix the constants hidden by the big-$O$ notation in some of our statements. We therefore henceforth denote the constant hidden by \Cref{cor:SSSPtoATO} to maintain the SSSP data structure by $c_{SSSP}$, the constant hidden in \Cref{lma:reductionATOtoSSSP} to obtain an $\mathcal{ATO}$-bundle from an SSSP data structure by $c_{SSSP \to ATO}$ and finally, we denote the constant hidden in the theorem that we want to prove by $c_{Total}$ where we require that $c_{Total} \geq (c_{SSSP \to ATO} \cdot c_{SSSP})^2 \cdot 2^{49}$.

Without further due, let us prove the theorem by induction on $n$, the number of vertices in graph $G$. The base case with $n \leq 1$ is easily established since there are no paths in a graph of only one vertex thus we obtain arbitrarily good quality and the running time is a small constant (at least smaller than $c_{Total}$).

Let us now give the inductive step $n \mapsto n+1$: for each $i$, we iteratively construct an $\mathcal{ATO}(G, 2^i, 40 c \log n)$-bundle $\mathcal{S}_i$ as described in \Cref{lma:reductionATOtoSSSP}. Thus, we have to show how to implement a $2$-approximate $2^{i-2}$-restricted $\mathcal{SSSP}$ data structure required in the reduction (note that for $i \leq 2$ this task is trivial, so we omit handling it as special levels).

Note that each data structure $\mathcal{SSSP}$ that we are asked to implement for an $\mathcal{ATO}$-bundle $\mathcal{S}_i$ at level $i$ is run on a different graph $H \subseteq G$. To obtain an efficent algorithm, we will implement the data structure differently depending on the size of such $H$. If $H$ has at least $n/2^\gamma$ vertices for some $\gamma = \Theta( \lg \log Wn)$ that we fix later, we call $H$ a \emph{large} graph. Otherwise, we say $H$ is \emph{small}. Now, we implement $\mathcal{SSSP}$ as follows:
\begin{itemize}
    \item if $H$ is \emph{small}, then we use the induction hypothesis, find a $\mathcal{ATO}$-bundle $\mathcal{S}'$ for $H$ and and invoke \Cref{cor:SSSPtoATO} on $\mathcal{S}'$. We note that we need to set the parameter that controls the failure probability for $\mathcal{S}'$ to $c \cdot 4 \log n$ to ensure that it succeeds with high probability (this is since $\mathcal{S}'$ only succeeds with probability polynomial in $|V(H)|$ which might be very small).
    \item if $H$ is \emph{large}, we exploit \Cref{rmk:LaminarFamilyOfGraphs} which states that when the reduction asks to maintain approximate distances on some graph $H \subseteq G$, it is sufficient to maintain distance estimates on any graph $F$ such that $H \subseteq F \subseteq G$ and in particular, it is ok to simply run on the entire graph $G$. Therefore, we simply use $\mathcal{S}_{i-1}$ in combination with \Cref{cor:SSSPtoATO} and maintain distances in $G$. 
\end{itemize}

Let us now analyze the total running time. We start by calculating the running time required by each level $i$ separately. For some fixed $i$, we have that by \Cref{lma:reductionATOtoSSSP} we have running time 
\begin{equation} \label{eq:totalCostBoundATO}
c_{SSSP \to ATO} \left(\sum_{j=0}^{\lceil\lg 2^{i-2} \rceil} \; \sum_{k=0}^{2^{j+3} c \log^2 n} T_{SSSP}(m_{j,k}, n/2^{j}, 2^{i-2}, 2) + n^2 \log^3 n\right)
\end{equation}
where $\sum_j m_{j,k} \leq 16c \cdot m \log^2 n$ for all $k$, to maintain $\mathcal{S}_i$.

Let us analyze the terms $T_{SSSP}(m_{j,k}, n/2^{j}, 2^{i-2}, 2)$ based on whether $j < \gamma$ or not (i.e. depending on how the SSSP data structure was implemented):
\begin{itemize}
    \item if $j < \gamma$: then, by the induction hypothesis, we require time at most 
        \begin{align*}
    \begin{split}
    & c_{Total} (c 4 \log n)^5 \cdot \left(\frac{n}{2^j}\right)^2 \log^{17}(n) \lg^{5}(Wn) 2^{-2j} \\
    & < c_{Total} \cdot 2^8 \cdot c^5 n^2 \log^{22}(n) \lg^{5}(Wn)  2^{-j}2^{-\gamma}
    \end{split}
    \end{align*}
    to maintain the new $\mathcal{ATO}$-bundle $\mathcal{S}'$ on the graph $H$ and again by the induction hypothesis we have that the bundle has quality $\frac{(4\log n c +2)40000n \log^5 n 2^{-j}}{2^{i-3}} \log^3 n$. 
    
    Thus, maintaining SSSP on $H$ using $\mathcal{S}'$ as described in \Cref{cor:SSSPtoATO} can be done in time 
    \[
     c_{SSSP}((c \cdot 4\log n+2)40000 \cdot 2^3)(n^2 \log^8 n 2^{-2j}) < c_{SSSP} \cdot c \cdot 2^{22} (n^2 \log^9 n 2^{-j} 2^{-\gamma}).
    \]
    Combined, we obtain that we can implement the entire SSSP data structure with total running time at most
    \begin{align*}
    \begin{split}
    &c_{Total} \cdot 2^8 \cdot c^5 n^2 \log^{22}(n) \lg^{5}(Wn) 2^{-j}2^{-\gamma} + c_{SSSP} \cdot c \cdot 2^{22} (n^2 \log^9 n 2^{-j}2^{-\gamma})\\
    &\leq c_{Total} \cdot c_{SSSP} \cdot c^5 \cdot 2^{-j}2^{-\gamma} \cdot 2^{22} \cdot n^2 \log^{22}(n) \lg^{5}(Wn).
    \end{split}
    \end{align*}

    Combining these bounds and summing over all \emph{small} graph terms in equation \ref{eq:totalCostBoundATO}, we obtain that the total contribution is at most
    \begin{align*}
    \begin{split}
    & c_{SSSP \to ATO} \cdot \lg(nW) \cdot (2^3 c\log^2 n) \left(c_{Total} \cdot c_{SSSP} \cdot c^5 \cdot 2^{-\gamma} \cdot 2^{22} \cdot n^2 \log^{22}(n) \lg^{5}(Wn)\right) \\
    &\leq c^6 \cdot c_{SSSP \to ATO} \cdot c_{Total} \cdot c_{SSSP} \left(2^{22} \cdot n^2 \log^{24}(n) \lg^{6}(Wn) 2^{-\gamma}\right).
    \end{split}
    \end{align*}
    This completes the analysis of the small graph data structures.
    \item otherwise ($j \geq \gamma$): then, we run the SSSP structure from \Cref{cor:SSSPtoATO} on $\mathcal{S}_{i-1}$ which gives running time at most
    \[
    c_{SSSP}\left(n 2^i \cdot \frac{(c+2)40000n \log^5 n}{2^{i-3}} \log^3 n\right) \leq c_{SSSP} \cdot c ( 2^{20} \cdot n^2 \log^8 n)
    \] 
    where we used $c \geq 2$. Since there are at most $c_{SSSP \to ATO} \cdot c \cdot (2^3 \lg(nW)\log^2 n 2^{\gamma})$ terms for large graphs, where $j \geq \gamma$, we have that the total cost of all $\mathcal{SSSP}$ data structures on \emph{large} graphs is at most
    \begin{align*}
    \begin{split}
   &c_{SSSP \to ATO} \cdot c \cdot (2^3 \lg(nW)\log^2 (n) 2^{\gamma}) \cdot \left(c_{SSSP} \cdot c ( 2^{20} \cdot n^2 \log^8 n) \right)\\
   &= c_{SSSP \to ATO} \cdot c^2 \cdot c_{SSSP} \cdot \left(2^{23} \cdot n^2 \log^{10} (n)  \lg(nW) 2^{\gamma}\right)
    \end{split}
    \end{align*}
\end{itemize}

It now only remains to choose $\gamma$ and combine the two bounds. We set $\gamma = 24 \lceil\lg (c^2 \cdot c_{SSSP \to ATO} \cdot c_{SSSP} \cdot \lg^{3}(Wn) \log^{7}(n)) \rceil$, and obtain that the total running time summed over large and small graphs is at most
\begin{align*}
\begin{split}
& c_{SSSP \to ATO}^2 \cdot c_{SSSP}^2 \cdot c^4 \cdot 2^{48}(n^2 \log^{17} n \lg^{4}(Wn))) \\ &+ \frac{c^4 c_{Total} (n^2 \log^{17} n \lg^{3}(Wn))}{2}
\\ &\leq  c^4 c_{Total} (n^2 \log^{17} n \lg^{4}(Wn))
\end{split}
\end{align*}
where we used our initial assumption on the size of $c_{Total}$.

Finally, we point out that there are at most $\lg(nW)$ levels $i$ and therefore, the total update time is at most
\[
 c_{Total} (c^4 n^2 \log^{17} n \lg^{5}(Wn))  
\]
as required.

Further, we point out that every $\mathcal{ATO}$-bundle $\mathcal{S}_i$ that was constructed runs correctly with high probability at least $1-n^{-c}$, while every $\mathcal{ATO}$-bundle $\mathcal{S}'$ is maintained correctly with probability at least $1-n^{-4c}$ (recall that we set the failure parameter of these data structures to $c \cdot 4 \log n$). Noting that we only have $\lg(nW)$ instances of the former bundles, and at most $n^3$ of the latter, taking a simple union bound over the events that any bundle instance fails gives a total failure probability of at most $n^{-c+2}$. 
\end{proof}

\section{A SSSP Algorithm for Sparse Graphs}
\label{sec:SSSPsparse}

Finally, we give a construction that is efficient in sparse graphs. Again, we prove \Cref{thm:ContributionSSSPSparseResult} in two steps: we first show how to reduce $\mathcal{SSSP}$ to maintaining an approximate topological order and then we bootstrap the approach.

\subsection{\texorpdfstring{$\alpha$-approximate $\delta$-restricted $\mathcal{SSSP}$}{SSSP} via Maintaining an Approximate Topological Order}

In this section, we give the following theorem that takes an $\mathcal{ATO}(G, \eta_{diam}$ and shows how to implement $(1+\epsilon)$-approximate $\delta$-restricted $\mathcal{SSSP}$. As in the previous section, we first prove a simpler theorem and then derive the $\mathcal{SSSP}$ data structure as a corollary.

\begin{restatable}{theorem}{ssspSparse}
\label{thm:SSSPEfficientSparse}
Given $G=(V,E,w)$, a decremental weighted digraph, a source $r \in V$, a depth threshold $\delta > 0$, a quality parameter $q$ such that $\delta q \geq n$, an approximation parameter $\epsilon > 0$, and access to $(\mathcal{V}, \tau)$ an $\mathcal{ATO}(G, \eta_{diam})$. 

Then, there exists a data structure that maintains a distance estimate $\widetilde{\mathbf{dist}}(r,v)$ for every vertex $v \in V$ such that at each stage of $G$, $\mathbf{dist}_G(r,v) \leq \widetilde{\mathbf{dist}}(r,v)$ and if $\mathbf{dist}_G(r,v) \leq \delta$ and $\mathcal{T}(\pi_{r,v},\tau) \leq q \cdot \delta + n$, then 
\[
\widetilde{\mathbf{dist}}(r,v) \leq (1+\epsilon)\mathbf{dist}_G(r,v) + \eta_{diam}.
\]
The expected total time required by this structure is 
\[
O((m\delta q/n^{-1/3} + \delta q n^{2/3}) \log^3 n \log Wn/\epsilon) 
\]
and the data structure runs correctly with probability at least $1 - n^{-c}$ for any constant $c > 0$. 
\end{restatable}

We point out that \Cref{rmk:ssspRemark} also applies to this theorem. As before, we use the theorem above to establish a simple corollary where we refer the reader to the proof of \Cref{cor:SSSPtoATO} which with minor changes proves the below result.

\begin{restatable}{corollary}{corSSSPSparserestrictedFromATO}
\label{cor:SSSPSparsetoATO}
Given $G=(V,E,w)$, a decremental weighted digraph, a source $r \in V$, a depth threshold $\delta > 0$, an approximation parameter $\epsilon > 0$, and access to a collection $\mathcal{S} = \{ S_i \}_{1 \leq i \leq \mu}$ for $\mu = \lfloor\lg \delta\rfloor-1$ where each $\mathcal{S}_i$ forms an $\mathcal{ATO}(G, 2^i, 40 c\log n)$-bundle of quality $q \cdot 2^i$. Then, there exists an implementation for $(1+\epsilon)$-approximate $\delta$-restricted $\mathcal{SSSP}$ where $T_{SSSP}(n,m,\delta, \epsilon) = O((m\delta q/n^{-1/3} + \delta q n^{2/3}) \log^3 n \log^2 Wn/\epsilon) $ with high probability.
\end{restatable}

As in the previous section, we defer the proof of \Cref{thm:SSSPEfficientSparse} to \Cref{sec:proofSSSPSparse} but present the implementation of a data structure $\mathcal{E}_r$ that satisfies the guarantees of \Cref{thm:SSSPEfficientSparse}. We therefore first show how to maintain a hopset $H$ of $G$ using $(\mathcal{V}, \tau)$ and then show that we can run an ES-tree on the hopset $H$ from source $s$. Before we discuss the internal parts of the data structure, let us define some useful concepts.

\paragraph{Preliminary Concepts.} In our algorithm, it is key to distinguish between the hop $h$ and the depth $\delta$ that was given. We therefore introduce the superscript $h$ to distances to denote the weight of the $h$-hop restricted shortest path in $G$ between two vertices $s,t \in V$ by $\mathbf{dist}_G^h(s,t)$. 

Further, we henceforth use $h_i = 2^i$ and we give data structures for every $i \leq \lg n$ that recovers paths from $r$ to vertices at distance at most $\delta$ and hop $h_i$ and assume for convenience that $n$ is a power of $2$ wlog. In order to describe the data structure for every hop level $i$, let us discuss some preliminary concepts.

As in equation \ref{eq:chi}, we also use the function $\chi_{close}$ defined by
\[
 \chi_{close}(X ,Y, \tau) \stackrel{\text{def}}{=}          \begin{cases}
        \tau(Y) - (\tau(X) + |X| - 1) & \text{if } \tau(X) < \tau(Y) \\
        \chi_{close}(Y, X, \tau) & \text{otherwise}
     \end{cases}
\]
Intuitively, $\chi_{close}$ takes two nodes $X,Y$ and maps to distance of the points in their intervals that are closest. We also set up a function $\chi_{far}$ that maps to the farthest distance between any such two numbers. We define
\[
 \chi_{far}(X ,Y, \tau) \stackrel{\text{def}}{=}          \begin{cases}
        \tau(Y) - \tau(X) + |Y| - 1 & \text{if } \tau(X) < \tau(Y) \\
        \chi_{far}(Y, X, \tau) & \text{otherwise}
\end{cases}
\]

Using these functions, we define the balls of vertices at small topological order distance.

\begin{definition}[Topological Order Difference Ball.]
For each $v \in V$ and integer $K \geq 0$, we define the closed topological order difference ball by
\[
    C_{closed}(v, K, ( \mathcal{V}, \tau)) \stackrel{\text{def}}{=}  \{ Y \in \mathcal{V} \; | \; \chi_{close}(X^v, Y, \tau) \leq K\}
\]
and the corresponding open ball by
\[
    C_{open}(v, K, ( \mathcal{V}, \tau)) \stackrel{\text{def}}{=}  \{ Y \in \mathcal{V} \; | \; \chi_{far}(X^v, Y, \tau) \leq K\}
\]
\end{definition}

We observe that $C_{closed}(v, K, ( \mathcal{V}, \tau))$ is a decreasing, refining set over time, i.e. new all sets are subsets of earlier sets\footnote{This is however not true for $C_{closed}(v, K, ( \mathcal{V}, \tau))$.}. Finally, let us give a definition for a \emph{decremental refining graph} which eases the description of the hopset.

\begin{definition}[Decremental Refining Graph]
\label{def:DecrRefiningGraph}
Given a refining partition $\mathcal{V}$ of the universe $V$, and a graph $H$ on the node set $\mathcal{V}$, we say that $H$ is a \emph{decremental refining graph} if at every stage, for any vertices $s,t \in V$, the distance between the nodes $X^s$ and $X^t$, the nodes containing the vertices $s$ and $t$ respectively, is monotonically increasing in $H$. In particular, any decremental graph $H'$ along with an arbitrary refining partition $\mathcal{V}'$ of its vertex set forms the decremental refining multi-graph $H' / \mathcal{V}'$.
\end{definition}

It is not hard to adapt the Generalized ES-trees presented in \cite{bernstein2019decremental}  combining it with a standard edge rounding technique for ES-trees (see for example \cite{bernstein2016maintaining}) in order to obtain the following theorem. We refer the reader to \cite{bernstein2019decremental} for a full proof.

\begin{lemma}[c.f. Lemma 4.2 and Lemma 4.3 in \cite{bernstein2019decremental}]
\label{thm:GEStreeOnCondensation}
Given a decremental refining graph $H$ with refining node set $\mathcal{V}$ part of an $\mathcal{ATO}$ $(\mathcal{V}, \tau)$, an approximation parameter $\epsilon > 0$, a depth threshold $\delta > 0$ and hop threshold $h_i$ for some $0 \leq h_i \leq \lg n$, a vertex $v \in V$, and a positive integer $K$. 

Let us define the graph $H_v$ of $v$ by $H_v \stackrel{\text{def}}{=} H\left[C_{closed}(v, K, ( \mathcal{V}, \tau))\right]$.

Then, there exists a deterministic data structure $\mathcal{GES}_v$ that can maintain the distance estimates $\widetilde{\mathbf{dist}}(X^v, Y)$ from $X^v$, the node in $\mathcal{V}$ containing $v$, to every node $Y \in C_{closed}(v, K, ( \mathcal{V}, \tau))$ such that $\mathbf{dist}_{G_v}(X^v, Y) \leq \widetilde{\mathbf{dist}}(X^v, Y)$ and
\[
\widetilde{\mathbf{dist}}(X^v, Y) \leq (1+\epsilon)\mathbf{dist}^{h_i}_{H_v}(X^v, Y).
\]
The algorithm runs in total update time \[
O(|E_H(C_{open}(v, K, (\mathcal{V}, \tau)))| h_i \log K \log Wn/ \epsilon + K \log K + \Delta)
\]
where $W$ is the largest integer weight in $H$ and the former term is the number of edges in $H$ that is incident to a node in $C_{open}(v, K, ( \mathcal{V}, \tau))$ throughout the entire course of the algorithm, and $\Delta$ is the number of edge weight increases for edges incident to $C_{open}(v, K, ( \mathcal{V}, \tau))$.
\end{lemma}

\paragraph{Hopset.} For different hop thresholds $h_i$, we describe how to maintain the hopset $H_i$ which forms a decremental refining graph on the node set $\mathcal{V}$. 

If $h_i < n^{2/3} \log n$, we let $H_i$ be the empty graph throughout the algorithm. Otherwise, we initialize the data structure to maintain the hopset $H_i$ by calculating $l_{i} = \frac{h_i }{n^{1/3}}$,  and sampling each vertex $v$ in $V$ into the set $S$ of sampled vertices with probability $\frac{3(c+6) \cdot \log n}{l_{i}}$. We run a data structure $\mathcal{GES}_s$ for every vertex $s \in S$, with hop $l_{i}$ on the graph $G / \mathcal{V}$ where we choose $K_i = \frac{q \cdot \delta}{n^{1/3}}$ as described in \Cref{thm:GEStreeOnCondensation}. 

We maintain the hopset $H_i$ so that for any nodes $X, Y \in \mathcal{V}$, there is an edge $(X,Y)$ if
\begin{enumerate}
    \item there exists vertices $s,t \in S$, with $s \in X$ and $t \in Y$, \emph{and}
    \item $\widetilde{\mathbf{dist}}(X, Y)$ is at most $l_{i}$.
\end{enumerate} 
If such an edge is in $H_i$, we assign it weight $\widetilde{\mathbf{dist}}(X, Y)$. It is straight-forward to see that we can maintain the required distance estimates $\widetilde{\mathbf{dist}}(X, Y)$ for the required depth using the data structure $\mathcal{GES}_s$ (technically there can be multiple $s'$ in $X$ but in this case their $\mathcal{GES}_{s'}$ data structures maintain the same distance estimates so there is no ambiguity in our notation). We observe that $H_i$ is a decremental refining graph which follows since distances in $H_s = (G / \mathcal{V})\left[C_{closed}(s, K, ( \mathcal{V}, \tau))\right]$ can only increase over time and only when the distance between the nodes containing $s$ and $t$ exceeds $l_i$ in this graph, we delete $(s,t)$ from $H_i$.

\paragraph{SSSP on the Hopset.} Finally, we can run from the root vertex $r$, the data structure $\mathcal{GES}_r$ as described in \Cref{thm:GEStreeOnCondensation} for $K = n$ on the graph $(G / \mathcal{V}) \bigcup (\cup_i H_i)$ with the given weights and $(\mathcal{V}, \tau)$ to depth $\delta$ and hop $n^{2/3} \log n$. This completes the description of the algorithm.

\subsection{Bootstrapping an Algorithm for Unrestricted Depth}
\label{subsec:bootstrapSparse}

We can now prove the theorem below which is a more detailed version of \Cref{thm:TopologicalOrderMaintenanceOverviewSparse}. Using \Cref{lma:reductionATOtoSSSP}, this immediately implies \Cref{thm:ContributionSSSPSparseResult}. Since the proof is quite similar to the proof of \Cref{thm:TopologicalOrderMaintenanceBootstrap}, we defer the proof to \Cref{sec:proofOfBootstrapSparse}.

\begin{restatable}{theorem}{bootstrapSparse}
\label{thm:TopologicalOrderMaintenanceBootstrapSparse}
For any $0 \leq i \leq \lg(Wn)$, given a decremental digraph $G=(V,E,w)$, we can maintain a hierarchy $\mathcal{S} =\{\mathcal{S}_i\}_i$ where each $\mathcal{S}_i$ is a $\mathcal{ATO}(G, 2^i, 40 c \log n)$-bundle of expected quality $n/2^i$. The algorithm runs in total expected update time $O(mn^{2/3} \log^{16} n \log^{3}(Wn))$ against a non-adaptive adversary and is correct with high probability. 
\end{restatable}

\section{Conclusion}
\label{sec:conclusion}

In this article, we gave the first near-optimal algorithm to solve the decremental SSSP problem on dense digraphs. Combined with the recent result by Gutenberg et al \cite{gutenberg2020incrSSSP}, this establishes an $\tilde{O}(n^2 \log^4 W/ \epsilon)$ complexity for the partially-dynamic SSSP problem in directed weighted graphs. Moreover, we gave a simple new technique to derive a data structure for sparse graphs using our framework, which runs in total time $\tilde{O}(mn^{2/3}\log^3 W/\epsilon)$ and thereby vastly improves over the best existing result by Probst Gutenberg and Wulff-Nilsen as recently shown in \cite{gutenberg2020decremental}. 

This substantial progress on the problem motivates the following two open questions:
\begin{itemize}
    \item Can we obtain a near-optimal algorithm for partially-dynamic SSSP in directed graphs for any sparsity? Recent work by Bernstein et al. \cite{bernstein2019decremental} has shown that this is at least possible for the simpler problem of maintaining single-source reachability.
    \item Can we derandomize our results or make them work against an adaptive adversary? Partial progress on this question has been made in  \cite{gutenberg2020incrSSSP} where the incremental SSSP data structure presented is already deterministic which gave the first deterministic improvement in directed graphs over the ES-tree. In decremental graphs, a recent result by Bernstein et al. \cite{detDiSSSP} gives deterministic total update time $O(n^{2+2/3})$ following a result by Probst Gutenberg and Wulff-Nilsen \cite{gutenberg2020decremental} that broke the $O(mn)$ bound randomized but against an adaptive adversary. However, this is far from the near-optimal bound achieved in this paper. We point out that even in the undirected, unweighted setting, the best data structure requires total update time $\tilde{O}(\min\{mn^{0.5+o(1)}, n^2\})$ \cite{gutenberg2020deterministic, bernstein2016deterministic} as opposed to $m^{1+o(1)}$ total update time in the non-oblivious setting \cite{henzinger2014sublinear}. 
\end{itemize}

\paragraph{Acknowledgements.} Aaron Bernstein is supported by NSF CAREER Grant 1942010 and the Simons Group for Algorithms \& Geometry. Maximilian Probst Gutenberg is supported by Basic Algorithms Research Copenhagen (BARC), supported by Thorup's Investigator Grant from the Villum Foundation under Grant No. 16582. Christian Wulff-Nilsen is supported by the Starting Grant 7027-00050B from the Independent Research Fund Denmark under the Sapere Aude research career programme. 
The authors thank Thatchaphol Saranurak for insightful comments and corrections, and anonymous FOCS reviewers for their helpful feedback.

\pagebreak

\printbibliography[heading=bibintoc] % Make bibliography show up in table of contents

\pagebreak

\appendix

\section{Related Work}
\label{sec:relatedWork}

We also point out that the simpler problem of Single-Source Reachability where the data structure only has to report whether there exists a path from the fixed vertex $s$ to some vertex $v \in V$ has been solved to near-optimality in \cite{bernstein2019decremental} which improved on the breakthrough results in \cite{roditty2008improved, lkacki2013improved, henzinger2014sublinear, henzinger2015improved, chechik2016decremental}. The algorithm in \cite{bernstein2019decremental} further even allows to maintain the Strongly-Connected Components of $G$. A recent result by Bernstein et al. \cite{detDiSSSP} further attempts to derandomize the algorithm in \cite{bernstein2019decremental} and deterministically achieves total update time $mn^{2/3+o(1)}$ to maintain Strongly-Connected Components, which constitutes the first deterministic improvement over the long-standing $O(mn)$ bound.

We further point out that the problem of maintain Single-Source Shortest Paths has also been considered quite recently in the incremental setting, i.e. the setting where a graph only undergoes edge insertions. In \cite{gutenberg2020incrSSSP}, the authors give a deterministic algorithm that achieves total update time $\tilde{O}(n^2 \text{polylog}(W))$. Unfortunately, the techniques proposed in this algorithm cannot be extended to the more interesting decremental setting.

There is also a wide literature on the related problems of All-Pairs Shortest Paths and All-Pairs Reachability both in the decremental \cite{baswana2007improved,bernstein2011improved, henzinger2014decremental, henzinger2016dynamic,bernstein2016maintaining, chechik2018near, karczmarz2020simple} and in the fully-dynamic setting \cite{henzinger1995fully, King99, demetrescu2001fully, demetrescu2004new, roditty2004dynamic, thorup2005worst,  roditty2012dynamic, abraham2013dynamic, roditty2016fully, henzinger2016dynamic, abraham2017fully, brand2019dynamic, probstWulffNilsenwcAPSP}.

Finally, for planar graphs, decremental algorithms are known to solve Single-Source Reachability deterministically in near-linear update time \cite{italiano2017decremental} and SSSP in directed graphs in total update time $\tilde{O}(n^{4/3})$ as shown by \cite{karczmarz2018decrementai}.

\section{Efficient Partitioning}
\label{sec:proofPartitionFull}

Let us now prove \Cref{lma:partitionFull} which is restated below for convenience.

\partition*

Let us start by claiming that \Cref{alg:split} is an efficient implementation of the procedure $\textsc{Partition}(G, d, \zeta)$ that gives the guarantees stated in \Cref{lma:partitionFull}. In the algorithm, we first initialize the separator set to the empty set and let $H$ throughout the algorithm be the graph $G$ without the vertices that are already contained in an SCC that satisfies the constraints. Thus, we run multiple iterations reducing the size of $H$ until it is empty and therefore all vertices satisfy our constraints.

\begin{algorithm}
\caption{$\textsc{Partition}(G, d, \zeta)$}
\label{alg:split}
\KwIn{A weighted digraph $G$, and integers $d$ and $\zeta$.}
\KwOut{Returns a set of edges $E_{Sep}$ such that for every SCC $X$ in $G \setminus E_{Sep}$, any two vertices $u,v \in X$, satisfy $\mathbf{dist}_{G \setminus E_{Sep}}(u,v) \leq d$.}
\BlankLine

$E_{Sep} \gets \emptyset; H \gets G;$\;

\While(\label{lne:outerwhile}){$H \neq \emptyset$}{
    Pick an arbitrary vertex $r$ in $V(H)$.\;
    Run $\textsc{OutSeparator}(r, H, d/8, 3\zeta \log n)$ and $\textsc{OutSeparator}(r, \overleftarrow{H}, d/8, 3\zeta  \log n)$, and let $(E'_{Sep}, V'_{Sep})$ be the tuple of the procedure such that $|E(V_{Sep})|$ is minimized.\label{lne:sepTwoWay}\;
    
    \If(\label{line:split-if-case}){$|V'_{Sep}| \leq \frac{2}{3}n$} {
        $E_{small} \gets \textsc{Partition}(H[V'_{Sep}], d, \zeta)$\label{lne:splitRecurseIf} \;
        $E_{Sep} \gets E_{Sep} \cup E_{Small} \cup E'_{Sep}$\label{lne:addtoS1}\;
        $H \gets H \setminus V'_{Sep}$
    }\Else(\label{line:split-else-case}){ 
        Initialize $\mathcal{A}_{r}$ to be a $2$-approximate $\delta$-restricted $\mathcal{SSSP}$ data structure on $H$\label{lne:BuildSSSP}\;
        \tcc{Find a good separator for every vertex that is far from $r$.}
         \While(\label{line:split-else-while}) {  
         there exists a vertex $v \in V(H)$ such that $\mathcal{A}_{r}$ has distance estimate $\widetilde{\mathbf{dist}}(r,v)$ or $\widetilde{\mathbf{dist}}(v,r)$ exceeding $\frac{d}{4}$}  {
            \If(\label{line:split-else-if}){$\widetilde{\mathbf{dist}}(v,r) > d/2$}{
                $(E''_{Sep}, V''_{Sep}) \gets \textsc{OutSeparator}(v, H, d/8, 3\zeta  \log n)$\;
            }\Else(\tcp*[h]{If $\widetilde{\mathbf{dist}}(r,v) > d/2$}){
                $(E''_{Sep}, V''_{Sep}) \gets \textsc{OutSeparator}(v, \overleftarrow{H}, d/8, 3\zeta  \log n)$\;
            }
            $H \gets H \setminus V''_{Sep}$\;
            
            $E'''_{Sep} \gets \textsc{Partition}(H[V''_{Sep}], d, \zeta)$ \label{lne:splitRecurse}\;
            $E_{Sep} \gets E_{Sep} \cup E''_{Sep} \cup E'''_{Sep}$\label{lne:addtoS2}\;
        }
        $H \gets \emptyset$\;
    }
}
\Return $E_{Sep}$\;
\end{algorithm}

In each iteration, we first pick an arbitrary vertex $r \in V(H)$ and then run the separator procedure described in lemma \ref{lma:sepIntro} on $H$ and on $\overleftarrow{H}$. We let $(E_{Sep}, V_{Sep})$ denote the tuple returned by the procedure where the size of $E(V_{Sep})$, that is the number of edges in $H$ incident to $V_{Sep}$, is minimized and break ties arbitrary if they are of equal size. We then check whether $E(V_{Sep})$ contains less than a $2/3$-fraction of the edges in $G$ in line \ref{line:split-if-case} (in fact, we could also compare to the number of edges in $H$, without affecting the asymptotic running time). In this case, we recurse on the graph $H[V_{Sep}]$, add the resulting edges to our edge set $E_{Sep}$ and remove the vertices $V_{Sep}$ from $H$. 

Otherwise, we know that most edges have their tail at small distance from and to $r$. We therefore initialize a $\mathcal{SSSP}$ data structure $\mathcal{A}_r$ from $r$ on $H$ in \Cref{lne:BuildSSSP}. Then, we successively extract a vertex $v$ that is far from $r$, compute a separator from/to them, and prune the set of vertices separated (after recursing on them). Since on termination of the while-loop in \Cref{line:split-else-while} each vertex is close to and from $r$, we have that all vertices are in an SCC of small diameter. Thus, we set $H = \emptyset$.

Let us now analyze the algorithm more formally. We start by showing that no induced subgraph that we recurse on in \Cref{lne:splitRecurse} can contain more than $\frac{2}{3} n$ vertices. This will be crucial to bound our running time, since it implies that every time that we recurse, we recurse on subgraphs that are significantly smaller, and therefore we make significant progress.

\begin{claim}
\label{clm:largeSCCifEStree}
If the algorithm enters \Cref{line:split-else-case}, then on termination of the while-loop starting in \Cref{line:split-else-while}, we have that the set vertices $V(H)$ in $H$ before it is set to $\emptyset$ is of size at least $\frac{1}{3}n$.
\end{claim}
\begin{proof}
Observe first that since we did not enter the if-case in \Cref{lne:splitRecurseIf}, that $|V'_{Sep}| > \frac{2}{3}n$. Since tuple $(E'_{Sep}, V'_{Sep})$ was chosen among the tuples $(E_{Sep}^{out}, V_{Sep}^{out})$ and $(E_{Sep}^{in}, V_{Sep}^{in})$ (the first corresponding to the first separator procedure, the second is the output of the second procedure in \Cref{lne:sepTwoWay} which are computed on $H$), we have
\[
|V^{out}_{Sep}| > \frac{2}{3}n \text{ and } |V^{in}_{Sep}| > \frac{2}{3}n.
\]
Further, we have by \Cref{lma:sepIntro} that $V^{out}_{Sep} \subseteq B^{out}_{H}(r, d/8)$ and $V^{in}_{Sep} \subseteq B^{in}_{H}(r, d/8)$. 

Thus, it is straight-forward to see that
\begin{equation}
\label{eq:manyEdgesInc}
    |B^{out}_{H}(r, d/8) \cap B^{in}_{H}(r, d/8)| \geq \frac{1}{3}n
\end{equation}

Finally, observe that during the while-loop starting \Cref{line:split-else-while}, we only prune vertices at distance at least $d/8$ in at least one direction. This is since distance estimates are overestimates so each vertex $v$ from which we start the procedure $\textsc{OutSeparator}(\cdot)$ is at distance at least $d/4$ from $r$ and by \Cref{lma:sepIntro} we only remove vertices that are at most at distance $d/8$ from $v$, so a straight-forward application of the triangle inequality implies our claim (a slightly subtle issue is that $H$ evolves, however distance estimates are with regard to the graph that $\mathcal{A}_r$ is initialized upon and $H$ has monotonically increasing distances over time so this might only help us). 

Since vertices in $B^{out}_{H}(r, d/8) \cap B^{in}_{H}(r, d/8))$ are close to $r$, none of them are pruned away. This is since distances do not change because even though we remove vertices from $H$, the $\mathcal{SSSP}$ data structure is run on the initial graph $H$, i.e. the one it was initialized upon. Thus, the vertices in the intersection of these balls remain in $H$ until the end. Combined with Equation \ref{eq:manyEdgesInc}, we derive the claim. 
\end{proof}

This claim is indeed the only crucial ingredient for our proof of \Cref{lma:partitionFull}, which we can now carry out.

\paragraph{Establishing Property \ref{prop:partitionDiameter}.} Let us prove the property by induction on the size of $E$. For the base case, assume $E = \emptyset$. Then, the claim is vacuously true. 

For $|E| = i + 1 > 0$, we observe that every time, we enter the if-case in line \ref{line:split-if-case} for some set $V'_{Sep}$, we add the edges $E'_{Sep}$ to $E_{Sep}$, thus for any vertex $u \in V \setminus V_{Sep}, v \in V_{Sep}$, we have that they are not strongly-connected in $G \setminus E_{Sep}$. Thus, we do not have to establish any guarantee for these pairs. We then recurse on $H[V_{Sep}] \subseteq G[V_{Sep}]$ which has fewer edges than $G$ by the if-condition. Thus, our claim for SCCs contained in this subgraph is true by the induction hypothesis.

Otherwise, we enter the else-case in \Cref{line:split-else-case}. By \Cref{clm:largeSCCifEStree}, on termination of the while-loop starting in \Cref{line:split-else-while}, the graph $H$ is still incident to a third of the edges in $G$. Thus, whenever we prune away a subgraph due to taking a separator chosen in \Cref{line:split-else-if}, we can again invoke the induction hypothesis when we recurse in \Cref{lne:splitRecurse}.

Thus, it only remains to establish on termination of the while-loop in \Cref{line:split-else-while}, the graph $H$ satisfies the property. Let $H'$ denote the graph that the data structure $\mathcal{A}_r$ was initialized on, so $H \subseteq H' \subseteq G$. It is clear that the vertices remaining in $H$ are at distance at most $d/2$ to and from $r$ (with regard to $H'$). Thus, for any two such vertices $x, y \in V(H)$, we have
\[
\mathbf{dist}_G(x,y) \leq \mathbf{dist}_{H'}(x,y) \leq \mathbf{dist}_{H'}(x,r) + \mathbf{dist}_{H'}(r,y) \leq d/2 + d/2 \leq d
\]
as desired.

\paragraph{Establishing Property \ref{prop:partitionProb}.} Let us first observe that we only add edges to the separator set $E_{Sep}$ in lines \ref{lne:addtoS1} and \ref{lne:addtoS2}. There are further four different places that emerge where an edge $e$ could have been added to a set that is then added to the separator set. We next observe that the edge $e$ is either added due to a recursive call or due to invoking procedure $(E''''_{Sep}, V''''_{Sep}) \gets \textsc{OutSeparator}(\cdot)$ (here $(E''''_{Sep}, V''''_{Sep})$ is a placeholder for the tuple returned in either one of the procedures). However, if $E''''_{Sep}$ is really added to $E_{Sep}$ and edge $e$ has its tail in $V''''_{Sep}$, then we have 
\begin{enumerate}
    \item that by \Cref{lma:sepIntro}, we added $e$ to $E_{Sep}$ with probability at most $\frac{8 \cdot 3\zeta \log n}{d}w(e)$, and
    \item that $V''''_{Sep}$ is removed from $H$ and since separator edges are computed in $H$, $e$ can afterwards only be added to $E_{Sep}$ due to a subsequent recursive call to $\textsc{Partition}(\cdot)$, and
    \item the graph $G[V''''_{Sep}]$ contains at most a $\frac{2}{3}$-fraction of the vertices in $G$. 
\end{enumerate}
Combining these facts, it is straight-forward to argue that the total probability of $e$ being added can be bound by $\log_{3/2} n \cdot \frac{24 \zeta \log n}{d}w(e) \leq \frac{240 \zeta \log^2 n}{d}w(e)$. 

\paragraph{Bounding the Error Probability.} First, observe that the probability of failure each procedure $\textsc{OutSeparator}(\cdot)$ in our algorithm is at most $e^{-3\zeta \log n}$ by \Cref{lma:sepIntro}. Now observe that every separator computed by this procedure that is used, separates two vertex sets that have so far been part of the same graph. Thus, throughout all subcalls, there can be at most $n-1$ such separators. Further observe, that if a separator computed by the procedure is not used, then we enter the else-case in \Cref{line:split-else-case}. Thus, this can only happen for each recursive call once. But since the number of recursive calls is exactly the number of used separators, this also occurs only $n-1$ times. Finally, we point out that in order to get a separator, we might run two separator procedures to get a single separator (namely in \Cref{lne:sepTwoWay}). Thus, the total number of separator procedures used throughout all the algorithm (including all subcalls) is at most $4(n-1)$. Using a simple union bound over all bad events, we can bound the probability of failure by $4n \cdot e^{-3\zeta  \log n} \leq e^{-\zeta}$.

\paragraph{Bounding the Running Time.} Let us first bound the running time of the procedure excluding recursive calls. We have that for each iteration of the while-loop in line \ref{lne:outerwhile}, the separator procedures in line \ref{lne:sepTwoWay} can be implemented efficiently, by running in parallel (that is they are executed such that the machine interleaves operations from the separator procedures). Then, if one of them terminates with tuple $(E'_{Sep}, V'_{Sep})$, it runs in time $O(|E(V'_{Sep})| \log n)$ by lemma \Cref{lma:sepIntro}. Further observe that if the other procedure runs longer than $O(|E(V'_{Sep})| \log n)$ by some large constant, we can abort it since it will not produce a tuple that is eligible for becoming $(E'_{Sep}, V'_{Sep})$ again by the running time guarantees of \Cref{lma:sepIntro}. Thus, line \ref{lne:sepTwoWay} can be implemented in time $O(|E(V'_{Sep})| \log n)$. If we enter the if-case in line \ref{line:split-if-case}, then we prune the set $E(V'_{Sep})$ from $H$. Thus, we can amortize all while-iterations entering the if-case over the edges removed from $H$, leading to total update time $O(|E| \log n)$. 

The final while-iteration might enter the else-case in line \ref{line:split-else-case}. Then, maintaining the data structure $\mathcal{A}_r$ takes total time $T_{SSSP}(n, m, \delta, 2)$ (here we could use $|V(H)| \leq n$ and $|E(H)| \leq m$, however this will not have significant impact). 

Finally, summing over the recursive calls where we use again the argument that the number of vertices in each induced subgraph is at most $\frac{2}{3}n$ by \Cref{clm:largeSCCifEStree}. We thus obtain total running time
\[
O\left(\sum_{j=0}^{\lceil\lg \delta \rceil} \; \sum_{k=0}^{2^{j+1}} T_{SSSP}(m_{j,k}, n/2^{j}, \delta, 2) + m \log^2 n\right)
\]
with $\sum_{k=0}^{2^{j+1}} m_{j,k} \leq 2m$. To derive the last bound observe that each vertex (and therefore incident edge) is in a single instance on each recursive level and after at most two recursive levels, the number of vertices has decreased by factor at least $2$.

\section{ \texorpdfstring{Proof of \Cref{lma:EStreeprob}}{Proof that the Number of While-Iterations is Small}}
\label{sec:EStreeprobProof}

\participation*

\begin{proof}
For $v \in V$, consider any two iterations of the while-loop in \Cref{lne:loop} where $v \in C$, say iteration $t_1$ and iteration $t_2$, where $t_1 < t_2$. Let $X^{v, t}$ refer to the SCC containing $v$ in $G'$ after the while-loop iteration $t$ ended, where $t_1 \leq t \leq t_2$.

We claim that with probability at least $1/2$, we have $|X^{v, t_2}| \leq \frac{1}{2} |X^{v, t_1}| $. This implies the theorem, since every time this event occurs the size of the SCC that $v$ is contained in is halved (observe that sequence of trials until one halving event occurs constitutes a geometric random variable with expectation $2$), and the SCC of $v$ can be halved at most $\lg n$ times.

Let us define by $G^t$ the graph $G$ after the $t^{th}$ iteration of the while loop and by $H^t$ the graph $G'$ after the $t^{th}$ iteration of the while-loop. Further let us define the graph $F^t$, for $t \geq t_1$, to be the defined $F^t = H^{t_1} \setminus (G^t \setminus G^{t_1})$. That is $F^t$ is the graph $G^{t_1}$ to which only \emph{adversarial} updates where applied after iteration $t_1$. So that is the edge deletions to $G'$ due to \emph{separator} edges are ignored after iteration $t_1$ and on.

Next, let $X_{max,t} \subseteq X^{v, t_1}$ denote some maximal set of vertices such that $H^{t}[X_{max,t}]$ is a SCC containing $v$ of diameter at most $\frac{\delta|X^{v,t_1}|}{16n}$. Similarly, let $Y_{max,t} \subseteq X^{v, t_1}$ denote some maximal set of vertices such that $F^{t}[Y_{max,t}]$ is a SCC containing $v$ of diameter at most $\frac{\delta|X^{v,t_1}|}{16n}$. Observe that $|X_{max,j}| \leq |Y_{max,j}|$ since $F^j \supseteq H^j$.

Next, let $t_1 \leq t'$ be the index such that,
\[
    |Y_{max,t_1}| < \frac{1}{2}|X^{v, t'}| \leq |Y_{max,t_1}|
\]
that is $t'$ is the first iteration after which $G'[X^{v, t_1}]$ has no SCC of diameter $\frac{\delta|X^{v,t_1}|}{16n}$ containing $v$ of size more than $\frac{1}{2}|X^{v, t_1}|$.

Now consider the case where some vertex in $Y_{max, t'-1}$ was chosen as a center $\textsc{Center}(X^{v,t_1})$ in iteration $t_1$. Then, observe that in each while-loop iteration, we always have a vertex $t$ that is at least at distance  $\frac{\delta|X^{v,t_1}|}{8n}$ from $\textsc{Center}(X^{v,t_1})$. Thus, the separator is taken at distance at least $\frac{\delta|X^{v,t_1}|}{16n}$ in \Cref{lne:DelSep} or \Cref{lne:DelSepRev}. Thus, we never remove a vertex at distance smaller-equal-than $\frac{\delta |X^{v,t_1}|}{16n}$ from $\textsc{Center}(X^{v,t_1})$ due to a separator procedure (and since we reuse the distance threshold until the component decreases by factor $2$ in size). Thus, if such a center was chosen, none of the vertices in $Y_{max, t'-1}$ would have participated in $C$ before iteration $t'$. Thus, if $t_2 < t'$, then the set $C$, with $v \in C$, was of size at most $\frac{1}{2}|X^{v, t_1}|$ and therefore the new SCC containing $v$ would have at most half the size of the previous one. One the other hand, if $t_2 \geq t'$, then $G'[C]$ would not contain a single induced SCC containing $v$ of diameter $\frac{\delta|X^{v,t_1}|}{16n}$ that is of size larger than $\frac{1}{2} |X^{v, t_1}|$. Since we invoke $\textsc{Partition}(\cdot)$ in that iteration on $G'[C]$, and by the guarantees given in \Cref{lma:partitionFull}, the SCC $X^{v, t_2}$ certainly satisfies $|X^{v, t_2}| \leq \frac{1}{2} |X^{v, t_1}|$.

Finally, we observe that since $\frac{1}{2}|X^{v, t_1}| \leq |Y_{max,t'-1}|$, we would have chosen a center in $Y_{max,t'-1}$, with probability at least $1/2$ initially. This completes the proof.
\end{proof}

\section{The \texorpdfstring{$\mathcal{SSSP}$}{SSSP} Data Structure}
\label{sec:proofSSSP}

Let us now show how to prove \Cref{thm:SSSPEfficient} which is restated below for convenience.

\ssspSimple*
\ssspRemark*

For a description of the algorithm, we refer the reader to \Cref{subsec:alphaDeltaSSSPReduction}.

\subsection{Correctness of the Algorithm}

We first point out that subsequently in the correctness proof, the term $\eta_{diam}$ actually appears in the lower bound instead of the upper bound on distance estimates. This however can be rectified by simply adding $\eta_{diam}$ to every distance estimate. Further, we observe that by property \ref{prop:TauTotal}, we obtain that $\mathcal{T}(\pi_{r,u},\tau)$ is upper bounded by $q\delta + n$, thus $\frac{\epsilon \mathcal{T}(\pi_{r,u},\tau)}{q} \leq \epsilon(\delta + n/q) \leq 2\epsilon\delta$ ($n/q$ is always upper bounded by $\delta$). Thus, we get the guarantees described in  \Cref{thm:SSSPEfficient} if we invoke the algorithm with $\epsilon' = \epsilon/2$ (which only affects the running time by a constant factor).

\begin{lemma}[Proof of Correctness]
\label{lma:correctnessOfEStree}
For each $u \in V$, we have that $\widetilde{\mathbf{dist}}(r,u)$ is a monotonically increasing distance estimate such that 
\begin{itemize}
    \item ${\mathbf{dist}}_G(r,u) \leq \widetilde{\mathbf{dist}}(r,u)$, at each stage, and
    \item if $\mathbf{dist}_G(r,u) \leq \delta$ and $\mathcal{T}(\pi_{r,u},\tau) \leq q \cdot \delta + n$, then $\widetilde{\mathbf{dist}}(r,u) \leq \mathbf{dist}_G(r,u) + \frac{\epsilon \mathcal{T}(\pi_{r,u} ,\tau)}{q} + \eta_{diam}$.
\end{itemize} 
\end{lemma}
\begin{proof}
In this proof, we use the superscript $t$ to refer to the version of a variable at the end of stage $t$ (it will be second if there already is a superscript indicating the node $X^x$ that some vertex $x$ is contained in, i.e. $X^{x,t}$). We recall that we maintain the distance estimate $\widetilde{\mathbf{dist}}(r,u)$ to be equal to $\widetilde{\mathbf{dist}}(X^r,X^u) + \eta_{diam}$ and observe that the distance estimates of $X^u$ and $u$ are monotonically increasing over time since $X^{u, t+1}$ is a subset of the set $X^{u, t}$, and if it is a strict subset, then its distance estimate was initialized to the distance estimate $\widetilde{\mathbf{dist}}^{t}(X^{r, t},X^{u, t})$ and afterwards, the update procedure can only have increase the estimate.

Now, let us establish that the distance estimate $\widetilde{\mathbf{dist}}^{t}(r,u) = \widetilde{\mathbf{dist}}^{t}(X^{r,t},X^{u,t}) + \eta_{diam}$, at any stage $t$, satisfies the stretch guarantee. We start by establishing the lower bound: we therefore observe that the certificate $T^t$ is formed by a subset of edges $E^t$, and the node set always corresponds to the partition $\mathcal{V}^t$. Thus, every path in $T$ must be at least of length equal to the shortest-path in $G^t / \mathcal{V}^t$. But since property \ref{prop:ContractLittle} enforces that any path has at most $\eta_{diam}$ of its length contracted, the lower bound follows.

We now prove the upper bound by induction on the distance $\mathbf{dist}^{t}(r,u)$ for every $u \in V$ where we assume that $\mathbf{dist}^t(r,u) \leq \delta + \eta_{diam}$ and $\mathcal{T}(\pi_{r,u, G^t},\tau^t) \leq q \cdot \delta + n$.
\begin{itemize}
    \item \underline{Base case $\mathbf{dist}^{t}(r,u) = \eta_{diam}$:} since edge weights are positive, this is only true if $u = r$. Since the distance estimate of $r$ is never changed by our update procedure (since $r$ never looses an incoming edge in $T$ and therefore is never added to $Q$), the distance estimate is $0$, as desired.
    \item \underline{Inductive step for $\mathbf{dist}^{t}(r,u) = d + 1 + \eta_{diam}$:} Let $x$ be the vertex in $\pi_{r,u,G^{t}}$ that precedes $u$ (i.e. the second last vertex on the shortest-path). Since $(x,u)$ is in $G$ at stage $t$, we have that at every stage $t' \leq t$, the edge $(x,u)$ is in $Q_{X^{x, (t')}, X^{u, (t')}}$ and since the function $\chi(X^{x, (t')}, X^{u, (t')})$ is monotonically increasing in $t'$, we therefore also have that $X^{x, (t')} \in B_{\leq j}(X^{u, (t')})$ for $j = \lfloor \lg \chi(X^{x, t}, X^{u, t}) \rfloor$ since our bucket update procedure always puts $X^{x, (t')}$ into the bucket $B_{j'}(X^{u, (t')})$ with the largest $j'$ and then $\chi(X^{x, t}, X^{u, t})$ only increases during the ensuing stages. 

    We retell that $X^{u, 0}$ had its distance estimate initialized to $\mathbf{dist}^{0}(X^r, X^u) \leq \mathbf{dist}^{t}(X^r, X^u) = d + 1$. Every time $t''$, it was increased ever since by at most $\lceil 2^{j} \cdot \frac{\epsilon}{q} \rceil$, we scanned $X^{x,t''}$ since it is in $B_{\leq j}(X^{u,t''})$. Thus, $X^{u,t}$'s distance estimate is at most $\frac{2^{j}\epsilon}{q} + w_{G^t}(x,u)$ larger than the distance estimate of $X^x$ (here we use that $w_{G^t}(x,u) \geq w_{G^{t'}}(x,u)$ for all $t' \leq t$). Therefore,
    \begin{align*} \widetilde{\mathbf{dist}}^{t}(X^{r, t},X^{u,t}) 
    &\leq \widetilde{\mathbf{dist}}^{t}(X^{r, t},X^{x,t}) + w_{G^t}(x,u) + \frac{2^{j}\epsilon}{q}\\
    &\leq \mathbf{dist}^{t}(r, x) + \eta_{diam} + \frac{\epsilon \mathcal{T}(\pi_{r,x,G^{t}} , \tau)}{q} + w_{G^t}(x,u) + \frac{2^{j}\epsilon}{q}\\ 
    &\leq \mathbf{dist}^{t}(r, x) + \eta_{diam} + \frac{\epsilon \mathcal{T}(\pi_{r,x,G^{t}} , \tau)}{q} + w_{G^t}(x,u) + \frac{2^{j}\epsilon}{q}\\
    &\leq (\mathbf{dist}^{t}(r, x) + \eta_{diam} + w_{G^t}(x,u)) + \frac{\epsilon \mathcal{T}(\pi_{r,x,G^{t}} , \tau)}{q} + \frac{\chi(X^{x, t}, X^{u, t}, \tau)\epsilon}{q}\\
    &\leq (\mathbf{dist}^{t}(r, x) + \eta_{diam} + w_{G^t}(x,u)) + \frac{\epsilon \mathcal{T}(\pi_{r,x,G^{t}} , \tau)}{q} + \frac{\mathcal{T}(X^{x, t}, X^{u, t}, \tau)\epsilon}{q}\\
    &= \mathbf{dist}^{t}(r, u) + \eta_{diam} + \frac{\epsilon \mathcal{T}(\pi_{r,x,G^{t}}, \tau)}{q} + \frac{\mathcal{T}((x,u), \tau)\epsilon}{q}\\
    &= \mathbf{dist}^{t}(r, u) + \eta_{diam} + \frac{\epsilon \mathcal{T}(\pi_{r,u,G^{t}}, \tau)}{q}
    \end{align*}
    where we use the induction hypothesis on $x$ in the second inequality, then plug in the value of $j$, observe that $\chi$ is dominated by $\mathcal{T}$ and finally observe that the function $\mathcal{T}$ is linear. We finally observe that by our assumptions  $\widetilde{\mathbf{dist}}^{t}(X^{r, t},X^{u,t}) \leq \mathbf{dist}^{t}(r, u) - \eta_{diam}  + \frac{\epsilon \mathcal{T}(\pi_{r,u,G^{t}}, \tau)}{q} \leq \delta + \frac{\epsilon (q\delta +n)}{q} = (1+\epsilon)\delta + \epsilon n/q \leq \delta_{max}$ and since it was monotonically increasing, there was no previous stage at which the distance estimate could have been reassigned $\infty$ due to exceeding $\delta_{max}$.
\end{itemize}
\end{proof}

\subsection{Running Time}

It remains to establish that the algorithm given above is efficient. Let us start by analyzing the bucket update procedure.

\begin{claim}
\label{clm:almostInTheRightBucket}
At any stage $t$, for any $X,Y \in \mathcal{V}$, with $2^{j} \leq \chi(X, Y, \tau) < 2^{j+1}$, the update procedure ensures that we have we have $X$ in $B_j(Y)$ or $B_{j-1}(Y)$. 
\end{claim} 
\begin{proof}
We observe that $X$ should be moved to a different bucket if $\chi(X,Y,\tau)$ changes by a significant amount. We observe further, that $\chi(X,Y, \tau)$ is monotonically increasing over stages and increases between two stages $t' \leq t''$ by at most $|X^{t'} \setminus X^{t''}| + |Y^{t'} \setminus Y^{t''}|$. This follows since $\chi(X,Y, \tau)$ maps to $\tau(Y) - (\tau(X) + |X| - 1)$ which increases if either $Y$ is mapped to a larger $\tau$-value due to a split or due to a lower size of $X$, again, due to a split. 

Now, let $t_{last}$ be the last stage that $X$ was reassigned to a bucket $B_{j'}(Y)$. Then, at that stage, $X$ was assigned according to the initialization rule, so $2^{j'} \leq \chi(X, Y, \tau) < 2^{j'+1}$. However, this implies that since $t_{last}$ neither $X$ nor $Y$ have decreased in size by $2^{j'}$ or more since otherwise the pair would be scanned (a simple proof by contradiction establishes this claim). Thus, we can upper bound the increase in $\chi(X, Y, \tau)$ since this stage by
\[
|X^{t} \setminus X^{t_{last}}| + |Y^{t} \setminus Y^{t_{last}}| < 2^{j'+1}
\]
and therefore we have that $\chi(X, Y, \tau) < 2^{j'+2}$.
\end{proof}

\begin{claim}
\label{clm:runtimeUpdateGraph}
The total time required to update the graph $\mathcal{H}$ and the corresponding buckets is $O(n^2 \log n)$.
\end{claim}
\begin{proof}
Every time we split a node $X$ into nodes $X_1, X_2, \dots$, the largest new node $X_i$ inherits the node of $X$, which can be implemented in constant time by reassigning pointers. Then, for each $X_{i'}$, we scan the incident edges to construct the new vertices. However, $X_{i'}$ is of size at most half the size of $X$. Thus, each vertex (and incident edge) participates only $O(\log n)$ times in such a node split. Therefore, the total time required by such scans is $O(m \log n) = O(n^2 \log n)$. Claim \ref{clm:almostInTheRightBucket} implies that \emph{after} the update procedure, each bucket $B_j(Y)$ contains at most $2^{j+2}$ vertices. Since each vertex can only be in $n/2^j$ nodes that scan $B_j(Y)$ due to the size decrease (or node versions at which a size update is triggered), and since there are only $O(\log n)$ values for $j$, we scan at most $O(n^2 \log n)$ pairs without moving a vertex to another bucket.

Further since each $X$ in $B_j(Y)$ can only be assigned to the same bucket or a bucket with larger index, each vertex $X$ can change $O(\log n)$ times the bucket at $Y$, thus we have at most $O(n^2 \log n)$ scans overall that where some node is moved to another bucket.

Since each scan only takes constant time, our claim now follows. 
\end{proof}

Finally, we establish that our procedure to reconstruct $T$ is efficient.

\begin{claim}
\label{clm:fixingTRuntime}
The total update time to construct and maintain the certificate $T$ is $O(n \log n \frac{\delta_{max} q}{\epsilon} )$.
\end{claim}
\begin{proof}
We first observe that we can initialize $T$ by Dijkstra's algorithm in time $O(|E(\mathcal{H})| \log |E(\mathcal{H})|) = O(n^2 \log n)$. On updates, for each distance value $\widetilde{\mathbf{dist}}(X^r, Y)$, a node $Y \in \mathcal{V}$ computes $j$ to be the largest integer such that the distance estimate is divisible by $\lceil 2^j \frac{\epsilon}{q}\rceil$. Thus, for a specific $j$, there are at most $\frac{\delta_{max}}{2^j} \cdot \frac{q}{\epsilon}$ distance values where $j$ is computed. We then scan at the distance value, the current set $B_j(Y)$ which is of size $O(2^{j})$ by \Cref{clm:almostInTheRightBucket}. Note that the node $Y$ might be in $Q$ multiple times until its distance value increases. However, every node $X$ in $B_j(Y)$, once scanned and not used to repair the certificate $T$ can be ignored for this distance value since distance estimates and edge weights are monotonically increasing, thus it can never be used at a later stage to repair $T$ at the current distance value. Thus, the total cost to scan the nodes in $B_j(Y)$ for all $j$ can be bounded by $O(\delta_{max} \cdot \frac{q}{\epsilon})$ and since we have at most $n$ nodes at any point and $O(\log n)$ buckets at each node, we have total cost $O(n \log n \frac{\delta_{max} q}{\epsilon} )$ for these scans. We also point out that the set $B_j(Y)$ might be changed between two fixing procedures of $T$ while the distance value of $Y$ has not changed. However, we can amortize the scans of new items over the bucket update procedure and obtain that at most $O(n^2 \log n)$ additional scans are necessary due to changes of $B_j(Y)$ for all nodes $Y$. 

Finally, we point out that $\frac{\delta_{max}q}{\epsilon} = \Omega(n)$ by definition and therefore the latter bound is subsumed.
\end{proof}

Using $\epsilon' = \epsilon/2$ in the algorithms, we can combine, \Cref{lma:correctnessOfEStree}, Claims \ref{clm:runtimeUpdateGraph} and \ref{clm:fixingTRuntime}, and the value of $\delta_{max}$ establish \Cref{thm:SSSPEfficient}. 

\section{\texorpdfstring{$\mathcal{SSSP}$ from $\mathcal{ATO}$-bundles}{SSSP from ATO-bundles}}
\label{sec:ssspfromATOfullProofCor}

In this section, we prove \Cref{cor:SSSPtoATO}.

\corSSSPrestrictedFromATO*

\begin{proof}
We run for each $(\mathcal{V}, \tau) \in \mathcal{S}_i$ a data structure as described in \Cref{thm:SSSPEfficient} from $r$ on $G$ and $\overleftarrow{G}$, respectively, with approximation parameter $\epsilon$ and depth threshold $\delta/(\epsilon 2^i)$. By \Cref{def:ATObundle} and the guarantees given in \Cref{thm:SSSPEfficient}, we have for any pair $(s,t) \in (\{r\} \times V) \cap (V \times \{r\})$, that some data structure has a distance estimate $\widetilde{\mathbf{dist}}(s,t)$ such that if $\delta/(\epsilon 2^{i-1}) \leq \mathbf{dist}(s,t) \leq \delta/(\epsilon 2^i)$, then
\[
\widetilde{\mathbf{dist}}(s,t) \leq (1+4\epsilon)\mathbf{dist}(s,t).
\]
Further, we have that each distance estimate is an overestimate of the actual distance. We observe that the only distance range uncovered is distances that are smaller than $1/\epsilon$. We therefore run a single instance of an ES-tree as given in \cite{shiloach1981line} which runs in time $O(n^2/\epsilon)$.

Thus, letting $\widetilde{\mathbf{dist}}(s,t)$ be the minimal value of any distance estimate, in any of the data structures, we have that our data structure maintains for each relevant tuple $(s,t)$, a $(1+4\epsilon)$ distance estimate (this can be done by maintaining a heap over all such distance estimates). Rescaling $\epsilon$ slightly, gives the desired result.

The total update time can be bound straight-forwardly since we run $O(\log^2 n)$ instances of the data structure in \Cref{thm:SSSPEfficient} where each is at cost
\[
O(n \frac{\delta}{\epsilon 2^i} \cdot q_i \log n/\epsilon + n^2 \log n) = O(n (\max_{1 \leq i \leq \mu} \{\frac{\delta q_i}{2^i}\} + n)\log n/\epsilon^2)
\]
and a single ES-tree at cost $O(n^2 /\epsilon)$. The costs of maintaining a heap over all distance estimates is easily subsumed by the running time.
\end{proof}

\section{Proof of \texorpdfstring{ \Cref{thm:SSSPEfficientSparse}}{the Sparse Graph SSSP Reduction}}
\label{sec:proofSSSPSparse}

In this section we are concerned with proving \Cref{thm:SSSPEfficientSparse}.

\ssspSparse*

We start the proof by showing a lemma that bounds construction time and size of $H_i$. We then proceed to prove for each level $i$ that we can find distance estimates for vertices at distance at most $\delta$ from the root vertex $r$ using the hopset $H_i$ and finally we combine the results to prove \Cref{cor:SSSPSparsetoATO}. Before we establish the first lemma, let us state a version of the Chernoff bound that we use frequently in our analysis.

\begin{theorem}[Multiplicative Chernoff Bound]
\label{thm:Chernoff}
Suppose $X_1, X_2, \dots, X_k$ are independent $O/1$-random variables. Let $X = \sum_{i=1}^k$ and let the expected value of $X$ be denoted by $\mu = E[X]$. Then for any $\delta \geq 1$, we have that
\[
    \mathbb{P}\left[ X > (1+\delta)\mu \right] \leq e^{-\frac{\delta\mu}{3}}.
\]
\end{theorem}

Equipped with this theorem, let us now start the analysis.

\begin{lemma}
\label{lma:timeSpaceHopset}
For every $0 \leq i < \lg n$, the hopset $H_i$ has at most $O(\delta q \log n)$ edges and running time $O(m \delta q\log n \log nW / (n^{1/3}\epsilon))$ to maintain the hopset $H_i$ with probability at least $1- O(n^{-(c+1)})$.
\end{lemma}
\begin{proof}
The claim is trivially true for $i$ where $h_i < n^{2/3} \log n$ since then $H_i$ is empty by definition. Otherwise, consider any vertex $v \in V$, then we have that $v$ is contained in a node in $C_{open}(w, K_i, ( \mathcal{V}, \tau)$ during the entire course of the algorithm of at most $2 K_i$ vertices $w \in V$. This follows since by definition of $\chi_{far}$ all nodes $Y$ in the set $C_{open}(w, K_i, ( \mathcal{V}, \tau)$ have their interval $[\tau(Y), \tau(Y) + |Y|)$ contained in the interval $[\tau(X^w) + |X^w| - K_i, \tau(X^w) + K_i)$. But this implies that there is only $2K_i$ vertices in such nodes.

It is not hard to establish that every set $C_{open}(w, K_i, ( \mathcal{V}, \tau)$ contains at most $\frac{18(c+6) K_i \log^2 n}{ l_i}$ vertices in $S$ throughout the entire course of the algorithm with probability at least $1-n^{-(c+2)}$ by a straight-forward Chernoff bound application (see \Cref{thm:Chernoff}). Using a union bound, we can upper bound the probability that \emph{any} vertex $w$ has more than $\frac{18(c+6) K_i \log^2 n}{ l_i}$ vertices in $S$ in the set $C_{open}(w, K_i, ( \mathcal{V}, \tau)$ by $1 - n^{-(c+1)}$. Henceforth in the proof, we condition on the event that no such event occurs.

Since, for every vertex $s \in S$, we have at most $\frac{18(c+6) K_i \log^2 n}{ l_i} + 2$ nodes throughout the algorithm for which an edge is inserted into $H_i$. The additive plus $2$ term stems from the fact that the sets $C_{open}(s, K_i, ( \mathcal{V}, \tau)$ and $C_{closed}(s, K_i, ( \mathcal{V}, \tau)$ differ by at most two nodes as pointed out before (the two nodes that have overlapping $\tau$-intervals with $[\tau(X^s) + |X^s| - K_i, \tau(X^s) + K_i)$). 

By a similar argument as above, there are at most $\frac{3(c+1) n \log n}{ l_i}$ vertices in $S$ with probability at least $1-n^{-(c+1)}$. Conditioning on the event that the number of vertices in $S$ does not exceed this bound, we can upper bound the number of edges in $H_i$ by 
\[
    \sum_{s \in S} |C_{closed}(s, K_i, ( \mathcal{V}, \tau)| \leq \frac{3(c+1) n \log n}{ l_i} \cdot \left(\frac{18(c+2) K_i \log^2 n}{ l_i} + 2\right) = O\left(\frac{K_i \cdot n \log^3 n}{l_i^2}\right).
\]
We recall that $l_{i} = \frac{h_i }{n^{1/3}}$ and $K_i = \frac{q \cdot \delta}{n^{1/3}}$, and therefore the total number of edges can be bound by $O\left(\frac{\delta q n^{1+1/3} \cdot \log^3 n}{h_i^2}\right)$ and since $h_i \geq n^{2/3} \log n$, we have at most $O(\delta q  \log n)$ edges in $H_i$.

Finally, let us bound the running time to maintain $H_i$. Using an similar argument as before, we have that every vertex $v \in V$ is only in the sets $C_{open}(s, K_i, ( \mathcal{V}, \tau)$ of at most $\frac{3(c+1) K_i \log n}{l_i}$ vertices $s \in S$ with probability $1-n^{-(c+1)}$. Conditioning on this event, we have that we have every edge in $G$, is contained in at most $\frac{6(c+1) K_i \log n}{l_i}$ sets $E_{G / \mathcal{V}}(C_{open}(s, K_i, ( \mathcal{V}, \tau))$ for $s \in S$. But then, the total update time to maintain all data structures $\mathcal{GES}_s$ can be bound by
\[
    O\left(m \cdot \frac{6(c+1) K_i \log n \log Wn /\epsilon}{l_i} \cdot l_i \right) = O(m \delta q\log n \log nW / (n^{1/3}\epsilon))
\]
It is not hard to verify that the time to maintain the hopset $H_i$ is subsumed in the running time of the data structures $\mathcal{GES}_s$ since after every time a node is split in such a data structure we need to can check in constant time whether the new nodes contain a vertex in $S$ and include a new edge, and again in case the distance to a node exceeds $l_i$, we can in constant time remove the corresponding edge if present in $H_i$. Edge weight changes further coincide with increases of distance estimates. We point out that the events that we did not condition upon occur with total probability at most $O(n^{-(c+1)})$.
\end{proof}

\begin{lemma}
\label{lma:hopboundHopset}
For any $0 \leq i < \lg n$, the graph $(G/\mathcal{V}) \cup H_i$ maintains for any pair of vertices $s,t \in V$ with $\mathbf{dist}_G(s,t) = \mathbf{dist}^{h_i}_G(s,t) \leq \delta$ and $\mathcal{T}(\pi_{s,t}), \tau) \leq q \delta + n$, that
\[
\mathbf{dist}^{n^{2/3} \log n}_{(G/\mathcal{V}) \cup H_i}(X^s,X^t) \leq (1+\epsilon)\mathbf{dist}_G(s,t).
\]
where $X^s$ and $X^t$ are the nodes in $\mathcal{V}$ that contain the vertices $s$ and $t$ respectively.
\end{lemma}
\begin{proof}
Clearly, for $h_i < n^{2/3} \log n$, we have that the claim is vacuously true since this implies that the path corresponding to $\mathbf{dist}^{h_i}_G(s,t)$ has less than $n^{2/3} \log n$ edges in $G$ and therefore the path in the graph $G / \mathcal{V}$ where some vertices are contracted into nodes can have at most the same hop.

For $h_i \geq n^{2/3} \log n$, let $\pi_{s,t}$ be a shortest path in $G$ between $s$ and $t$ of hop at most $h_i$. Let us partition the path into segments $\pi_1, \pi_2, \dots, \pi_k$ of exactly $l_i/2$ edges where only the last segment is possible smaller. We observe that $k \leq \lceil 2h_i / l_i \rceil \leq n^{2/3} + 1$. We claim that with exception of the last segment, each segment contains some vertex $v$ that was sampled into $S$. We can prove that this is indeed true by using a Chernoff bound as described in \cref{thm:Chernoff} where for each vertex in a fixed segment we have the experiment $X_i$ that is $1$ if the vertex is sampled into $S$ and $0$ otherwise. It is not hard that for a fixed segment, there is at least one such vertex with probability $1-n^{-(c+6)}$. Taking a union bound over the segments, we have that with probability at least $1-n^{-(c+5)}$ every segment (except the last one) contains a vertex in $S$. 

Next, let us select such a vertex $s_j \in S$ for each segment $\pi_j$. We observe that any two such vertices $s_j$ and $s_{j+1}$ in consecutive segments have a shortest path consisting of at most $l_i$ edges between them by the optimal substructure property of shortest paths, in particular they have the shortest path $\pi^i_{s,t}[s_j, s_{j+1}]$ with at most $l_i$ edges. 

We further observe that since $s_j$ has a $\mathcal{GES}_s$ data structure, by the argument above, we have that there is an edge $(X^{s_j} , X^{s_{j+1}})$ in $H_i$ of weight at most $(1+\epsilon)\mathbf{dist}_G(s_j, s_{j+1})$ if the shortest path $\pi^i_{s,t}[s_j, s_{j+1}]$ is contained in the graph $H_s$. We remind the reader that the graph $H_s$ was defined to be the graph $(G/ \mathcal{V})\left[C_{closed}(s, K, ( \mathcal{V}, \tau))\right]$. 

It is straight-forward to observe that for every segment between two vertices $s_j$ and $s_{j+1}$ we have 
\begin{itemize}
    \item either the edge $(X^{s_j} , X^{s_{j+1}})$ of weight at most $(1+\epsilon)w(\pi^i_{s,t}[s_j, s_{j+1}])$, or
    \item some vertex $v \in \pi^i_{s,t}[s_j, s_{j+1}]$ that is not in $C_{closed}(s, K_i, ( \mathcal{V}, \tau))$. 
\end{itemize}
Now, we observe that the latter case we have by definition of $C_{closed}(s, K_i, ( \mathcal{V}, \tau))$ that $|\tau(X^v) - \tau(X^s)| \geq K_i = \frac{q \cdot \delta}{n^{1/3}}$. But since we have by the definition of quality, that $\mathcal{T}(\pi_{s,t}), \tau) \leq q \delta + n$, we have that there can be at most $\frac{q \delta + n}{q \cdot \delta/n^{1/3}} = 2n^{1/3}$ such indices $j$ by the pigeonhole principle (here we use $\delta q \geq n$).

Finally, we can put everything together: we have that the shortest path from $s$ to $s_1$ contains at most $l_i/2$ edges, analogously for the last vertex $s_{k-1}$ to $t$. In between for each segment, we either have the direct edge in the hopset or we can take at most $l_i$ edges in $G / \mathcal{V}$ whenever we do not. Since this can occur at most $n^{1/3}$ times, we have that the path constructed by this method has hop at most
\[
    2 \cdot l_i + k + 2 n^{1/3} l_i \leq 4 n^{2/3}
\]
since $l_i \leq n^{1/3}$ and $k \leq n^{2/3}$ for all choices of $i$. This proves the lemma.
\end{proof}

\begin{proof}[Proof of \Cref{thm:SSSPEfficientSparse}]
It is now rather straight-forward to observe that the data structure $\mathcal{GES}_r$ maintains $(1+\epsilon)^2$-approximate distance estimates from $X^r$ to every node $Y \in \mathcal{V}$ that contains a vertex $t$ at distance at most $\delta$ from $s$ in $G$ and with $\mathcal{T}(\pi_{r,t}), \tau) \leq q \delta + n$ by combining  \Cref{thm:GEStreeOnCondensation} and \Cref{lma:hopboundHopset}. Rescaling $\epsilon$ by a constant factor, we obtain again a $(1+\epsilon)$-approximation on the distance estimates without affecting the running time asymptotically. 

Due to contractions according to $\mathcal{V}$, we might slightly underestimate the distances. However, since $(\mathcal{V}, \tau)$ form an $\mathcal{ATO}(G, \eta_{diam})$, we underestimate by at most $\eta_{diam}$ and adding $\eta_{diam}$ to every distance estimate remedies this problem without violation our constraints on the upper bound for the approximation. 

We further have that the running time of the data structure $\mathcal{GES}_r$ is at most the number of edges in $H_i$ and $G$. But by combining \Cref{thm:GEStreeOnCondensation} and \Cref{lma:timeSpaceHopset}, we obtain that the running time is at most
\[
    O((m + \sum_i \delta q \log n) n^{2/3} \log n \log Wn/ \epsilon) = O((m+\delta q) n^{2/3} \log^3 n \log Wn/\epsilon).
\]
Since we can maintain each $H_i$ in time $O(m \delta q\log n \log nW / (n^{1/3}\epsilon))$ by \Cref{lma:timeSpaceHopset}, we can bound the total update time that the data structure spends by 
\begin{align*}
O(m \delta q \log^2 n \log nW / (n^{1/3}\epsilon) + (m+\delta q) n^{2/3} \log^3 n \log Wn/\epsilon) \\= O((m\delta q/n^{-1/3} + \delta q n^{2/3}) \log^3 n \log Wn/\epsilon)
\end{align*}
where we use again $\delta q \geq n$.

Finally, we argue that each hopset $H_i$ is correctly maintained with probability $1-O(n^{-(c+1)})$ and therefore by a simple union bound, we can deduce that the overall algorithm runs correctly with probability at least $1- n^{-c}$. The theorem follows.
\end{proof}

\section{Proof of \texorpdfstring{ \Cref{thm:TopologicalOrderMaintenanceBootstrapSparse}}{the Bootstrapping Theorem for Sparse Graphs}}
\label{sec:proofOfBootstrapSparse}

In this section we prove \Cref{thm:TopologicalOrderMaintenanceBootstrapSparse}.

\bootstrapSparse*

\begin{proof}
We prove the theorem by induction on $n$, the number of vertices in graph $G$. The base case with $n \leq 1$ is easily established since we can subsume any term using a large enough constant.

Let us now give the inductive step $n \mapsto n+1$: we maintain a hierarchy of $\mathcal{ATO}$'s, for levels $0 \leq i \leq \lg(Wn)$. At level $i \leq 2$, we let $\mathcal{S}_i$ be a $\mathcal{ATO}(G, 2^i, 40 c \log n)$-bundle where each $(\mathcal{V},\tau) \in \mathcal{S}_i$ is an $\mathcal{ATO}(G, 2^i)$ obtained from letting $\mathcal{V}$ be the trivial partition of $V$ into singletons and by letting $\tau$ be an arbitrary permutation of the elements in $\mathcal{V}$. It is straight-forward to see that because the maximal topological difference is $n$, that each $(\mathcal{V}, \tau)$ has quality $n \leq \frac{(c+2)40000n \log^5 n}{2^2}$. 

For each $i > 2$, we then iteratively construct an $\mathcal{ATO}(G, 2^i, 40 c \log n)$-bundle $\mathcal{S}_i$ as described in \Cref{lma:reductionATOtoSSSP}. Thus, we have to show how to implement a $2$-approximate $2^{i-2}$-restricted $\mathcal{SSSP}$ data structure required in the reduction. 

We implement the $\mathcal{SSSP}$ data structure depending on the size of the graph that it is run upon, where we call a graph \emph{large} if it has at least $n/2^\gamma$ vertices and \emph{small} otherwise. It is here that we exploit \Cref{rmk:LaminarFamilyOfGraphs} which states that when the reduction asks to maintain approximate distances on some graph $H \subseteq G$, it is sufficient to maintain distance estimates on any graph $F$ such that $H \subseteq F \subseteq G$ and in particular, it is ok to simply run on the entire graph $G$.

Now, if the graph is \emph{large}, we run the data structure given in \Cref{cor:SSSPSparsetoATO} on the full graph $G$ using $\mathcal{S}_{i-3}, \mathcal{S}_{i-3}, \dots, \mathcal{S}_{0}$: thus each such data structure exploits the access to the existing $\mathcal{ATO}$s. If the graph is \emph{small}, then we construct a new hierarchy $\mathcal{S}' =\{\mathcal{S}'_i\}_i$ on the graph by invoking the induction hypothesis and then run an unrestricted $2$-approximate SSSP data structure derived by using \Cref{cor:SSSPSparsetoATO} on $\mathcal{S}'$. This completes the description of level $i$. 

Let us now bound the costs to maintain a single level $i$. For $i \leq 2$, we can bound the total running time by $O(n \log n)$ since we only have to publish $\mathcal{S}_0$ initially.

For $i > 2$, we have running time 
\begin{equation}
O\left(\sum_{j=0}^{\lceil\lg \delta \rceil} \; \sum_{k=0}^{2^{j+3} c \log^2 n} T_{SSSP}(m_{j,k}, n/2^{j}, \delta, 2) + m \log^3 n\right)
\end{equation}
where $\sum_j m_{j,k} \leq 16c \cdot m \log^2 n$ for all $k$, as given by \Cref{lma:reductionATOtoSSSP}.    

For each term $T_{SSSP}(m_{j,k}, n/2^{j}, 2^j, 2)$, where the graph is \emph{large}, i.e. $j \leq \gamma$, then we have 
\[
T_{SSSP}(m_{j,k}, n/2^{j}, 2^j, 2) = O(m n^{2/3} \log^8 n \log^2 Wn/\epsilon)
\] 
by \Cref{cor:SSSPSparsetoATO}. Since we have less than $O(\log^4 n 2^{\gamma})$ terms, where $j \geq \gamma$, we have that the total cost of all $\mathcal{SSSP}$ data structures on \emph{large} graphs is at most
\[
    O(2^{\gamma} m n^{2/3} \log^{12} n \log^2 Wn/\epsilon).
\]

Further, using the induction hypothesis, we have that every term $T_{SSSP}(m_{j,k}, n/2^{j}, \delta, 2)$ where $j < \gamma$, requires time $O(m_{j,k} n^{2/3} \log^{16} n \log^3 Wn/\epsilon)$. Summing over all such \emph{small} graph terms, we obtain total time
\[
O( 2^{-2/3 \cdot \gamma} m \log^{18} n \log^{3}(Wn)).
\]
Finally, we balance both terms by setting $\gamma = \Theta( \lg(\log^2(Wn)\log^4 n))$ with a sufficiently large constant to obtain total running time $O(mn^{2/3} \log^{16} n \log^{3}(Wn))$. Thus, summing over $\log(Wn)$ values for $i$, the total running time is established (slightly better trade-off values can be achieved but in order to keep bounds simple we use this trade-off).

Further, we point out that every $\mathcal{ATO}$ that was constructed runs correctly with high probability $1-n^{-c'}$ for constant $c' > 0$, and since there are only polynomially many instances, we can set $c'$ large enough to ensure that the entire algorithm runs correctly with probability $1- n^{-c}$ for any $c > 0$ by taking a union bound over the events that an instance fails. 
\end{proof}

\end{document}